\documentclass[11pt, letterpaper]{article}

\usepackage[utf8]{inputenc}
\usepackage[english]{babel}

\usepackage[dvipsnames,usenames]{xcolor}

\usepackage[margin=1in]{geometry}

\usepackage{amsthm,thmtools}
\declaretheorem{theorem}
\declaretheorem[sibling=theorem]{lemma}

\declaretheorem[style=definition]{definition}
\declaretheorem[style=definition, numbered=no, name=Example (Search Ads)]{search-ads}

\definecolor{algbg}{gray}{.9}
\declaretheoremstyle[
    postheadspace=\newline,
    notefont=\bfseries,
    shaded={bgcolor=algbg}
]{algorithm}
\declaretheorem[style=algorithm, name=Algorithm]{algorithm}

\usepackage{mathtools}
\usepackage{bm}

\usepackage{enumitem}

\usepackage[numbers,sort]{natbib}

\usepackage[colorlinks=true,pdfpagemode=UseNone,urlcolor=RoyalBlue,linkcolor=RoyalBlue,citecolor=OliveGreen,pdfstartview=FitH]{hyperref}

\usepackage{booktabs}
\usepackage{tikz}

\newcommand{\opt}{\textsc{Opt}}
\newcommand{\alg}{\textsc{Alg}}
\newcommand{\natlp}{\textsc{Nat}}
\newcommand{\jllp}{\textsc{JL}}

\newcommand{\emp}{\perp}

\newcommand{\unweighted}{0.711}
\newcommand{\vertexweighted}{0.7009}

\newcommand{\set}[1]{\left\{ #1\right\}}

\newcommand{\pdif}[2]{\frac{\partial{#1}}{\partial{#2}}}
\newcommand{\pdifsqr}[2]{\frac{\partial^2{#1}}{\partial{#2}^2}}

\allowdisplaybreaks[4]

\newcommand*{\defeq}{\stackrel{\mathrm{def}}{=}}

\title{Online Stochastic Matching, Poisson Arrivals, and the Natural Linear Program}

\author{
Zhiyi Huang
\thanks{The University of Hong Kong. Email: \href{mailto:zhiyi@cs.hku.hk}{zhiyi@cs.hku.hk}.}
\and
Xinkai Shu
\thanks{The University of Hong Kong. Email: \href{mailto:xkshu@cs.hku.hk}{xkshu@cs.hku.hk}.}
}

\date{March 2021}

\begin{document}

\begin{titlepage}
\thispagestyle{empty}
\maketitle
\begin{abstract}
    \thispagestyle{empty}
    We study the online stochastic matching problem.
Consider a bipartite graph with offline vertices on one side, and with i.i.d.\ online vertices on the other side.
The offline vertices and the distribution of online vertices are known to the algorithm beforehand.
The realization of the online vertices, however, is revealed one at a time, upon which the algorithm immediately decides how to match it.
For maximizing the cardinality of the matching, we give a $\unweighted$-competitive online algorithm, which improves the best previous ratio of $0.706$.
When the offline vertices are weighted, we introduce a $\vertexweighted$-competitive online algorithm for maximizing the total weight of the matched offline vertices, which improves the best previous ratio of $0.662$.

Conceptually, we find that the analysis of online algorithms simplifies if the online vertices follow a Poisson process, and establish an approximate equivalence between this Poisson arrival model and online stochstic matching.
Technically, we propose a natural linear program for the Poisson arrival model, and demonstrate how to exploit its structure by introducing a converse of Jensen's inequality.
Moreover, we design an algorithmic amortization to replace the analytic one in previous work, and as a result get the first vertex-weighted online stochastic matching algorithm that improves the results in the weaker random arrival model.

\end{abstract}

\bigskip


    %
    %
    %
    %
    %
    %

\end{titlepage}

\section{Introduction}
\label{sec:introduction}

Building on three decades of research started by \citet*{KarpVV:STOC:1990}, online matching has developed to be a central topic in the literature of online algorithms.
Among other applications, online advertising has been a main driving force behind this development.

\begin{search-ads}
    Consider a search engine.
    Advertisers want their ads to be shown to the users who search for certain keywords.
    When a user performs a search, the search engine needs to immediately pick an advertiser interested in the search term and show its ad to the user.
\end{search-ads}

This problem is often modeled as an online bipartite matching problem.
The vertices on one side correspond to the advertisers, and are known upfront.
We call them the \emph{offline vertices}.
The vertices on the other side correspond to the searches by users, and are revealed one at a time.
We call them the \emph{online vertices}.
The edges represent if the advertisers are interested in the search terms.
If the advertisers pay the same amount, say, $1$ cent, per display of their ads to the relevant searches, it is an \emph{unweighted} matching problem whose goal is to maximize the cardinality of the matching.
If different advertisers pay different amounts per display, it is a \emph{vertex-weighted} matching problem in which we aim to maximize the total weight of the matched offline vertices. 

\paragraph{Worst Case Model.}
\citet*{KarpVV:STOC:1990} considered the worst case model, which measures an online algorithm's performance in the worst graph and worst arrival order of the online vertices.
Concretely, for any online algorithm, consider the ratio of the expected size of the algorithm's matching to the maximum matching in hindsight, in the worst graph and arrival order that minimize the ratio.
This is called the \emph{competitive ratio}.
In this model, \citet*{KarpVV:STOC:1990} introduced the Ranking algorithm that achieves the optimal $1-\frac{1}{e} \approx 0.632$ competitive ratio in the unweighted problem.
\citet*{AggarwalGKM:SODA:2011} generalized it to the vertex-weighted problem.

\paragraph{Random Order Model.}
Subsequently, researchers found the competitive ratios from the worst case model to be too pessimistic, and introduced stochasticity to obtain better results.
The weakest form of stochasticity is the \emph{random order model}, which still considers the worst graph for any give algorithm but assumes that the online vertices arrive in a random order.
\citet*{MahdianY:STOC:2011} proved that the competitive ratio of Ranking for unweighted matching improves to $0.696$ in this model, and \citet*{KarandeMT:STOC:2011} showed that it is at best $0.727$-competitive.
For the vertex-weighted case, \citet*{HuangTWZ:TALG:2019} proposed a generalization of Ranking that exploits the random arrival order, and its competitive ratio was improved to $0.662$ by \citet*{JinW:arXiv:2020}.

\paragraph{Online Stochastic Matching.}
This paper will focus on the \emph{online stochastic matching} model, which makes a stronger stochastic assumption that the online vertices are independently and identically distributed (i.i.d.) according to a distribution.
The distribution is known to the algorithm, but the realization of the online vertices is not.
An online algorithm's competitive ratio is defined against the worst distribution.
In the unweighted case of this model, \citet*{FeldmanMMM:FOCS:2009} first beat $1-\frac{1}{e}$ competitive ratio, under the assumption of integral arrival rate. Without this assumption, \citet*{ManshadiOS:MOR:2012} gave the first algorithm, and the state-of-the-art is the $0.706$-competitive algorithm by \citet*{JailletL:MOR:2014}.
In the vertex-weighted online stochastic matching, however, there has been no improvement over the random order model, unless we make extra assumptions.
See Subsection~\ref{sec:related-work} for further related work on special cases of online stochastic matching.

\subsection{Our Contributions}
\label{sec:contribution}

We introduce new online algorithms to obtain improved competitive ratios in both the unweighted and the vertex-weighted problems, from $0.706$ to $\unweighted$ and from $0.662$ to $\vertexweighted$ respectively.
Our vertex-weighted algorithm and the analysis are the first in the literature that successfully exploit the stronger stochasticity in online stochastic matching than in the random order model.

\begin{table}[h]
    \centering
    \caption{A summary of the results in this paper and in previous work}
    \label{tab:summary}
    \medskip
    \begin{tabular}{ccc}
        \toprule
        & Unweighted & Vertex-weighted \\
        \midrule
        Worst Case Model & $1-\frac{1}{e} \approx 0.632$~\cite{KarpVV:STOC:1990} & $1-\frac{1}{e} \approx 0.632$~\cite{AggarwalGKM:SODA:2011} \\
        Random Order Model & $0.696$~\cite{MahdianY:STOC:2011} & $0.662$~\cite{JinW:arXiv:2020} \\
        Online Stochastic Matching & $0.706$~\cite{JailletL:MOR:2014} & $0.662$~\cite{JinW:arXiv:2020} \\
        \textbf{Online Stochastic Matching (This Paper)} & $\mathbf{\unweighted}$ & $\mathbf{\vertexweighted}$ \\
        \bottomrule
    \end{tabular}
\end{table}

\paragraph{Conceptual Contribution: Poisson Arrivals.}
We find that the competitive analysis of online algorithms become easier in a variant of online stochastic matching in which the online vertices follow a Poisson process.
In other words, the number of online vertices in this model is drawn from a Poisson distribution instead of a fixed number as in the original model.
For example, the asymptotic independence among some events in the analysis of \citet*{JailletL:MOR:2014} becomes genuine independence in the Poisson arrival model.
Furthermore, we show that for a natural family of online algorithms, their competitive ratios in the Poisson arrival model also apply to the original online stochastic matching model.
See Section~\ref{sec:preliminary} for detail.

\paragraph{Technical Contribution 1: Natural Linear Program.}
Similar to the previous work on online stochastic matching, we compare the algorithm's matching to an upper bound of the optimal given by a linear program (LP).
To this end, we consider arguably the most natural LP that one could write for this problem.
In fact, the LPs used by the previous work are all relaxations of this natural LP.
See Appendix~\ref{app:lp} for a comparison.
Although the natural LP has exponentially many constraints, we give a polynomial-time separation oracle and thus demonstrate its computational tractability.
Moreover, from the LP's constraints we derive a converse of Jensen's inequality, which is repeatedly used throughout the paper.
Section~\ref{sec:lp} presents this natural LP and its properties.

\paragraph{Technical Contribution 2: Algorithmic Amortization.}
The previous online algorithms for unweighted online stochastic matching rely on an amortized analysis.
For each offline vertex, we can decompose its probability of being matched by the algorithm into two parts, which we shall refer to as the \emph{basic} and \emph{extra} parts.
Instead of comparing the contribution of an offline vertex to the LP and the probability that it gets matched, i.e., the sum of its \emph{basic} and \emph{extra} parts, the amortized analysis considers the sum of its \emph{basic} part and its contribution to the \emph{extra} parts of the other vertices.
It fails in the vertex weighted-case because the contribution to the other vertices' \emph{extra} parts could be negligible if their weights are much smaller.
We overcome this obstacle in Section~\ref{sec:vertex-weighted} by moving from the analytic amortization to an algorithmic one.
When an online vertex samples an offline vertex to matched to, we let it drop the sampled offline vertex with some probability and let it resample, even if the offline vertex is not yet matched.
The drop rates are carefully designed based on how much the offline vertices are matched in the natural LP.
See Section~\ref{sec:vertex-weighted} for detail.

\subsection{Other Related Work}
\label{sec:related-work}

Besides the aforementioned results~\cite{ManshadiOS:MOR:2012, JailletL:MOR:2014}, online stochastic matching has also been studied in the special case of integral arrival rates, i.e., when the expected number of online vertices of each type is an integer.
In fact, when \citet*{FeldmanMMM:FOCS:2009} first introduced online stochastic matching, they focused on this case and gave a $0.67$-competitive algorithm.
Their algorithm is non-adaptive:
its matching decisions are independent of what happened in previous rounds.
Later, the competitive ratio was improved in a series of works.
\citet*{BahmaniK:ESA:2010} modified the algorithm of \citet{FeldmanMMM:FOCS:2009} to make it $0.699$-competitive.
\citet*{ManshadiOS:MOR:2012} proposed a $0.705$-competitive adaptive algorithm.
They also showed that no algorithm is better than $1-\frac{1}{e^2} \approx 0.862$-competitive even in the special case of unweighted matching with integral arrival rates, and for general arrival rates no algorithm is better than $0.823$-competitive.
Further, \citet*{JailletL:MOR:2014} designed LP-based algorithms that are $0.725$-competitive and $0.729$-competitive in the vertex-weighted and unweighted problems respectively.
Although the constraints in their LP are looser than those in our natural LP, their constraints exploit the integral arrival rates to ensure a semi-integral optimal solution.
As a result, it is easier to convert their LP solution into an online algorithm.
\citet*{Brubach:Algorithmica:2020} proposed a $0.7299$-algorithm by considering a different LP, which is between the LP of \citet{JailletL:MOR:2014} and ours in terms of the tightness of constraints.
See Appendix~\ref{app:lp} for a comparison of the LPs.

In the more general edge-weighted problem, and still under the assumption of integral arrival rates, \citet*{HaeuplerMZ:WINE:2011} proposed a $0.667$-competitive algorithm, and \citet*{Brubach:Algorithmica:2020} gave an improved $0.705$-competitive algorithm.

The broader online matching literature is too vast to be covered extensively.
Besides the mentioned results in the unweighted case~\cite{KarpVV:STOC:1990} and vertex-weighted case~\cite{AggarwalGKM:SODA:2011}, the edge-weighted case was studied by \citet*{FeldmanKMMP:WINE:2009} and \citet*{FahrbachHTZ:FOCS:2020}.
The algorithms and analysis have been unified under the online primal dual framework~\cite{DevanurJK:SODA:2013, DevanurHKMY:TEAC:2016}.
Other online matching problems from online advertising include AdWords~\cite{MehtaSVV:JACM:2007, BuchbinderJN:ESA:2007, GoelM:SODA:2008, DevanurH:EC:2009, HuangZZ:FOCS:2020} and online matching with stochastic rewards~\cite{MehtaP:FOCS:2012, MehtaWZ:SODA:2014, HuangZ:STOC:2020, GoyalU:EC:2020}.
See the survey by \citet{Mehta:FTTCS:2013} for further references.

\section{Online Stochastic Matching and Poisson Arrivals}
\label{sec:preliminary}

Consider the matching in a bipartite graph.
The offline vertices on one side are fixed.
The online vertices on the other side are i.i.d.
Let $I$ be the set of online vertex types.
Let $J$ be the set of offline vertices.
For any online type $i \in I$ and any offline vertex $j \in J$, let $w_{ij} \ge 0$ be the weight of matching an online vertex of type $i$ to the offline vertex $j$.
The problem is \emph{unweighted} if $w_{ij} \in \{ 0 , 1\}$, and is \emph{vertex-weighted} if $w_{ij} \in \{ 0, w_i \}$, for any $i \in I$ and any $j \in J$.
Each online type $i \in I$ further has arrival rate $\lambda_i$, which equals the expected number of online vertices of type $i$ in the graph.

\paragraph{Online Stochastic Matching.}
Online stochastic matching considers a random bipartite graph $G$ with $\Lambda = \sum_{i \in I} \lambda_i$ online vertices\footnote{In online stochastic matching setting the sum $\Lambda$ is an integer, while in the Poisson arrival model it could be any positive real number.} arriving one at a time on one side, and with offline vertices $J$ on the other side.
Each online vertex independently draws its type $i \in I$ with probability $\frac{\lambda_i}{\Lambda}$.
The set of online types and the corresponding weights $w_{ij}$'s and arrival rates $\lambda_i$'s are known to the algorithm, but the realization of the graph is not.

\paragraph{Poisson Arrival Model.}
The competitive analyses of online algorithms substantially simplify in a variant of the online stochastic matching model.
Instead of having a fixed number of online vertices, let each type independently follow a Poisson process with arrival rate $\lambda_i$.
Equivalently, draw the number of online vertices from a Poisson distribution with mean $\Lambda$.



\paragraph{Online Algorithms.}
An online algorithm makes the matching decision for each online vertex irrevocably and immediately upon its arrival.
Let $\alg$ be the expected total weight of the edges in the algorithm's matching.
We shall consider the standard competitive analysis with respect to the expected total weight of the maximum weight matching of the realized graph $G$, denoted as $\opt$.
The \emph{competitive ratio} of an online algorithm is the infimum of $\frac{\alg}{\opt}$ over all possible instances.

\begin{theorem}
    \label{thm:opt-compare}
    Fix any distribution of online vertices and any $\Lambda$:
    \begin{enumerate}
        \item The optimal of online stochastic matching is at least the optimal of the Poisson arrival model.
        \item The optimal of the Poisson arrival model is at least $1 - O(\Lambda^{-\frac{1}{2}})$ times the optimal of online stochastic matching.
    \end{enumerate}
\end{theorem}

\begin{proof}
    For the fixed distribution, let $\opt_n$ be the difference between the optimal of online stochastic matching with $n$ and $n-1$ online vertices.
    Since dropping a random online vertex from the optimal solution with $n$ vertices gives a solution to the case of $n-1$ vertices, $\opt_n$ is nonincreasing in $n$. 

    By definition, the optimal of online stochastic matching with the given $\Lambda$ equals $\sum_{n=1}^\Lambda \opt_n$.
    Similarly, the optimal of the Poisson arrival model equals $\sum_{m=1}^\infty \frac{\Lambda^m e^{-\Lambda}}{m!} \sum_{n=1}^m \opt_n$.
    Since both have $\Lambda$ vertices in expectation, the first part of the theorem follows from the monotonicity of $\opt_n$.

    Next we prove the second part.
    The optimal of the Poisson arrival model is lower bounded by:
    \begin{align}
        \sum_{m=1}^\infty \frac{\Lambda^m e^{-\Lambda}}{m!} \sum_{n=1}^{\min \{ \Lambda, m\}} \opt_n
        &
        \ge \sum_{m=1}^\infty \frac{\Lambda^m e^{-\Lambda}}{m!} \frac{\min \{ \Lambda, m\}}{\Lambda} \sum_{n=1}^{\Lambda} \opt_n
        \tag{monotone $\opt_n$}
        \\
        &
        = \Big( 1 - O \big( \Lambda^{-\frac{1}{2}} \big) \Big) \sum_{n=1}^{\Lambda} \opt_n
        \tag{tail bound of Poisson}
        ~.
    \end{align}

    We include a proof of the tail bound in Appendix~\ref{app:proof-poisson-tail} for completeness.
\end{proof}

\paragraph{Monotone Online Algorithms.}
For any $n \ge 1$, let $\alg_n$ denote the expected weight that the algorithm gets from matching the $n$-th online vertex.
An online algorithm is \emph{monotone} if $\alg_n$ is nonincreasing in $n$;
it is $\alpha$-\emph{approximately monotone} if $\alg_n \le \alpha \cdot \alg_\ell$ for any $n > \ell$.
Intuitively, natural online algorithms shall be monotone since there are fewer remaining offline vertices as $n$ increases.
Indeed, our unweighted algorithm is monotone, and our vertex-weighted algorithm is $O(1)$-approximately monotone.
To our knowledge, so are the existing algorithms in the literature.

\begin{theorem}
    \label{thm:alg-compare}
    Fix any distribution of online vertices and any $\Lambda$:
    \begin{enumerate}
        \item For any monotone algorithm, its objective in online stochastic matching is at least its objective in the Poisson arrival model.
        \item For any $\alpha$-approximately monotone algorithm, its objective in online stochastic matching is at least $1 - O(\alpha \Lambda^{-\frac{1}{2}})$ times its objective in the Poisson arrival model.
        %
    \end{enumerate}
\end{theorem}

\begin{proof}
    \textbf{(Part 1: Monotone Algorithms)~}
    By definition, the objective when there are $m$ online vertices equals $\sum_{n=1}^m \alg_n$.
    Hence, its objective in online stochastic matching is $\sum_{n=1}^\Lambda \alg_n$, and its objective in the Poisson arrival model is $\sum_{m=1}^\infty \frac{\Lambda^m e^{-\Lambda}}{m!} \sum_{n=1}^m \alg_n$.
    Since both have $\Lambda$ vertices in expectation, the first part of the theorem follows from the monotonicity of $\alg_n$.%
    \\[1ex]
    \textbf{(Part 2: Approximately Monotone Algorithms)~}
    The difference between the algorithm's objectives in the Poisson arrival model and in online stochastic matching is:
    \[
        \sum_{m=1}^\infty \frac{\Lambda^m e^{-\Lambda}}{m!} \sum_{n=1}^m \alg_n - \sum_{n=1}^\Lambda \alg_m = \sum_{n=1}^\infty \alg_n \sum_{m=n}^\infty \frac{\Lambda^m e^{-\Lambda}}{m!} - \sum_{n=1}^\Lambda \alg_n
        ~.
    \]

    Since $\sum_{m=n}^\infty \frac{\Lambda^m e^{-\Lambda}}{m!} < \sum_{m=0}^\infty \frac{\Lambda^m e^{-\Lambda}}{m!} = 1$, we can drop all $\alg_n$ for $1 \le n \le \Lambda$ and bound it by:
    \begin{align*}
        \sum_{n=\Lambda+1}^\infty \alg_n \sum_{m=n}^\infty \frac{\Lambda^m e^{-\Lambda}}{m!}
        &
        \le \sum_{n=\Lambda+1}^\infty \Big( \frac{\alpha}{\Lambda} \sum_{\ell=1}^\Lambda \alg_\ell \Big) \sum_{m=n}^\infty \frac{\Lambda^m e^{-\Lambda}}{m!}
        \tag{$\alpha$-approx.\ monotone $\alg_n$}
        \\
        &
        = \frac{\alpha}{\Lambda} \sum_{m=\Lambda+1}^\infty \frac{\Lambda^m e^{-\Lambda}}{m!} (m-\Lambda) \sum_{\ell=1}^\Lambda \alg_\ell \\
        &
        = O \big( \alpha \Lambda^{-\frac{1}{2}} \big) \sum_{\ell=1}^\Lambda \alg_\ell
        \tag{tail bound of Poisson}
        ~.
    \end{align*}

    Rearranging terms proves the theorem.
    See Appendix~\ref{app:proof-poisson-tail} for a proof of the tail bound.
\end{proof}

Hence, we will analyze our monotone or $O(1)$-approximately monotone algorithms in the Poisson arrival model.
By Theorem~\ref{thm:opt-compare} and Theorem~\ref{thm:alg-compare}, the competitive ratios then hold in both models, up to a $1 - O(\Lambda^{-\frac{1}{2}})$ factor which is negligible for sufficiently large instances.
We remark that previous works by \citet{ManshadiOS:MOR:2012} and \citet{JailletL:MOR:2014} also assumed sufficiently large instances.

\section{Natural Linear Program}
\label{sec:lp}

We consider the following LP relaxation, and let $\natlp$ denote its optimal value.
\begin{equation}
    \label{eqn:natural-lp}
    \tag{\natlp}
    \begin{aligned}
        \text{maximize} \quad & \sum_{i \in I} \sum_{j \in J} w_{ij} x_{ij} \\
        \text{subject to} \quad &
        \sum_{j \in J} x_{ij} \le \lambda_i && \forall i \in I \\
        &
        \sum_{i \in S} x_{ij} \le 1 - \exp \Big( - \sum_{i \in S} \lambda_i \Big) && \forall j \in J, \forall S \subseteq I \\
        & x_{ij} \ge 0 && \forall i \in I, \forall j \in J
    \end{aligned}
\end{equation}

It is \textit{natural} in the sense that the second constraint holds naturally in the Poisson arrival model.
Although it does not hold in the general case of online stochastic matching, it is asymptotically true when there are sufficiently many online vertices, which is focal case in existing works such as \citet{ManshadiOS:MOR:2012} and \citet{JailletL:MOR:2014}.
In deed, both works used the constraint $x_{ij} \le 1 - \frac{1}{e}$ in their LPs which is a special case of the second constraint in our natural LP. See Appendix~\ref{app:lp} for further discussions.

\begin{theorem}
    \label{lem:natural-lp}
    In the Poisson arrival model, $\opt \le \natlp$.
\end{theorem}

\begin{proof}
    We will construct a feasible solution to the natural LP whose objective equals the expected size of the optimal matching of the realized graph.
    Let $x_{ij}$ be the probability that offline vertex $j$ is matched to an online vertex of type $i$.
    Then, the objective of the natural LP equals the expected total weight of the matching.

    It remains to show feasibility.
    For any online type $i$, $\sum_{j \in J} x_{ij}$ is the expected number of matched online vertices of type $i$, which is no more than the expected number of online vertices of type $i$, i.e., $\lambda_i$.
    For any offline vertex $j$, and any subset of online types in $j$'s neighborhood $S \subseteq I$, $\sum_{i \in S} x_{ij}$ is the probability that $j$ is matched to an online vertex whose type is in $S$, which is no more than the probability that there is an online vertex whose type is in $S$, i.e., $1 - \exp \big( - \sum_{i \in S} \lambda_i \big)$.
    Finally for any $i \in I$ and any $j \in J$, $x_{ij}$ is nonnegative by definition.
\end{proof}

\subsection{Computational Tractability}

The natural LP has an exponential number of constraints.
Nonetheless, this subsection shows how to solve it in polynomial time using a separation oracle and the ellipsoid method.
We first introduce an equivalent form of the second constraint.

\begin{lemma}
    \label{lem:natural-lp-structure}
    The second constraint of the natural LP is equivalent to the following condition.
    For any offline vertex $j$, and any non-negative weights $0 \le \mu_i \le \lambda_i$ for $i \in I$:
    \[
        \sum_{i \in I} \frac{\mu_i x_{ij}}{\lambda_i} \le 1 - \exp \Big(- \sum_{i \in I} \mu_i \Big)
        ~.
    \]
\end{lemma}

\begin{proof}
    On the one hand, the second constraint of the natural LP is the special case of the condition in the lemma when $\mu_i \in \{ 0, \lambda_i \}$ for all $i \in I$.

    On the other hand, $\sum_{i \in I} \frac{\mu_i x_{ij}}{\lambda_i} + \exp \big(- \sum_{i \in I} \mu_i \big)$ is convex in $\mu_i$ for each $i \in I$.
    Hence, its maximum is achieved at a vertex of the feasible hyperrectangle, i.e., $\mu_i \in \{0, \lambda_i \}$ for all $i \in I$.
    In other words, the special case is sufficient for ensuring the general case.
\end{proof}

\begin{theorem}
    \label{thm:natural-lp-compuation}
    The natural LP is solvable in polynomial time.
\end{theorem}

\begin{proof}
    It suffices to find a separation oracle, in particular, for the second constraint of the natural LP.
    To do so, we propose an algorithm that for each offline vertex $j$ finds a subset of its neighborhood $S \subseteq I$ that maximizes:
    \[
        \sum_{i \in S} x_{ij} + \exp \Big( - \sum_{i \in S} \lambda_i \Big)
        ~.
    \]

    By Lemma~\ref{lem:natural-lp-structure}, this is equivalent to finding $0 \le \mu_i \le \lambda_i$ for $i \in I$ that maximizes:
    \[
        \sum_{i \in I} \frac{\mu_i x_{ij}}{\lambda_i} + \exp \Big(- \sum_{i \in I} \mu_i \Big)
        ~.
    \]

    For any fixed value of $\sum_{i \in I} \mu_i$ and thus the second term, the first term $\sum_{i \in I} \frac{\mu_i x_{ij}}{\lambda_i}$ is maximized when we assign $\mu_i$ greedily in descending order of $\frac{x_{ij}}{\lambda_i}$.
    Hence, the algorithm sorts $j \in J$ in descending order of $\frac{x_{ij}}{\lambda_j}$, and checks the constraint only for subsets $S$ comprised of the first $k$ elements in that order for $1 \le k \le |I|$.
\end{proof}

\subsection{A Converse of Jensen's Inequality}

For any convex function $f$, Jensen's inequality asserts that for any $j \in J$ (recall that $\Lambda = \sum_{i \in I} \lambda_i$):
\[
    \sum_{i \in I} \lambda_i f \Big( \frac{x_{ij}}{\lambda_i} \Big) \ge \Lambda f \Big( \frac{\sum_{i \in I} x_{ij}}{\Lambda} \Big)
    ~.
\]

On the other hand, the constraints of the natural LP bound how wide-spread the mass could be, leading to a converse of Jensen's inequality.

\begin{lemma}\label{lem:converse-jensen-general}
    For any convex function $f$ satisfying $f(0) = 0$, any offline vertex $j \in J$, and any feasible assignment $\big( x_{ij} \big)_{(i, j) \in E}$ of the natural LP:
    \[
        \sum_{i \in I} \lambda_i f \big( \frac{x_{ij}}{\lambda_i} \big) \le \int_0^{-\ln(1-x_j)} f \big( e^{-\lambda} \big) d \lambda
        ~.
    \]
\end{lemma}

\begin{proof}
    We will prove a more general result.
    Let $G(\lambda) = \min \big\{ x_j, 1 - e^{-\lambda} \big\}$ and let its derivate be:
    \[
        g(\lambda) = \begin{cases}
            e^{-\lambda} & \text{ if } \lambda \le - \ln (1 - x_j) ~; \\
            0 & \text{ otherwise.}
        \end{cases}
    \]
    
    Further consider an arbitrary differentiable $H$ such that $H(\lambda) \le G(\lambda)$ for all $\lambda \ge 0$; let $h$ denote its derivative.
    We claim that:
    \begin{equation}
        \label{eqn:converse-jensen-general}
        \int_0^\infty f \big( h(\lambda) \big) d\lambda \le \int_0^{-\ln(1-x_j)} f \big( e^{-\lambda} \big) d\lambda
        ~,
    \end{equation}
    where equality holds when $H = G$.
    Assume without loss of generality that $I = \{1, 2, \dots, |I|\}$ and $\frac{x_{ij}}{\lambda_i}$ is nonincreasing in $i$.
    The lemma follows as a special case when:
    \[
        H(\lambda) = \begin{cases}
            \sum_{i=1}^{k-1} x_{ij} + \frac{\lambda - \sum_{i=1}^{k-1} \lambda_i}{\lambda_{k}} x_{kj} &
            \text{ if } \sum_{i=1}^{k-1} \lambda_i \le \lambda < \sum_{i=1}^k \lambda_i \text{ for some $1 \le k \le |I|$;} \\
            x_j &
            \text{ if } \lambda \ge \sum_{i \in I} \lambda_i ~.
        \end{cases}
    \]

    Next we prove the general inequality in Eqn.~\eqref{eqn:converse-jensen-general}:
    \begin{align*}
        \int_0^\infty f \big( h(\lambda) \big) d\lambda
        &
        = \int_0^\infty \int_0^{h(\lambda)} f'(y) dy d\lambda \\
        &
        = \int_0^\infty \int_0^{h(\lambda)} \int_0^y f''(z) dz dy d\lambda \\
        &
        = \int_0^1 \big( H \big( h^{-1}(z) \big) - z h^{-1}(z) \big) f''(z) dz
        && \text{(change order of integration)} \\
        &
        \le \int_0^1 \big( G \big( h^{-1}(z) \big) - z h^{-1}(z) \big) f''(z) dz
        ~.
        && \text{($H(\lambda) \le G(\lambda)$)}
    \end{align*}

    Since $G(y) - zy$ is a concave function of $y$ and its derivative equals $0$ when $g(y) = z$, the maximum is achieved when $g(h^{-1}(z)) = z$, i.e., if $G = H$ and $g = h$.
\end{proof}

\section{Meta Algorithm}

This section presents a meta algorithm and establishes its properties.
It captures the algorithms in this paper, and the algorithms by \citet{ManshadiOS:MOR:2012} and \citet{JailletL:MOR:2014} as special cases. 

Upon the arrival of an online vertex, sample a pair of neighbors $(j,k)$ from a distribution that depends on its type $i$, independent to the sampled pairs for previous online vertices.
Then try $j$ as the first option.
If $j$ is already matched, continue to try $k$ as the second option.
We further define a dummy neighbor $\emp$, which will always be treated as already matched.
Hence, we may drop the first or the second option by letting $j = ~\emp$ or $k = ~\emp$.
Let $J^*=J \cup \set{\emp}$ be the extended set of offline vertices.

Formally, the algorithm is parameterized by a collection of distributions $D_i = \Delta \left( J^* \times J^* \right)$ for all $i \in I$.
Let $D_i(j, k)$ denote the probability of sampling $(j, k)$ from $D_i$.
See Algorithm \ref{alg:pair-sampling-unweighted}.

\begin{algorithm}[Pair Sampling]
    \label{alg:pair-sampling-unweighted}
    For each online vertex coming, say, of type $i$:
    \begin{enumerate}[itemsep=0pt, topsep=2pt]
        \item Sample $(j,k)$ from $D_i$.
        \item Match $i$ to $j$ if $j \neq \emp$ and it is not yet matched.
        \item Otherwise, match $i$ to $k$ if it is not yet matched.
    \end{enumerate}
\end{algorithm}

\subsection{Extended Types and Independence Properties}

We extend the type $i$ of an online vertex to be a tuple $(i,j,k)$ if the meta algorithm samples $(j, k)$.
Further, we say that an online vertex has type $(i,*,*)$ if its type is $(i,j,k)$ for some $j,k$, and likewise for types $(*,j,*)$ and $(*,j,k)$.
Let $\mu_{jk} = \sum_{i \in I} \lambda_i D_i(j,k)$ be the expected number of online vertices for which the algorithm samples pair $(j,k)$, for any $j, k \in J^*$.
Similarly, let $\mu_j=\sum_{k\in J^*}\mu_{jk}$ be the expected number of online vertices for which the algorithm samples $j$ as the first entry, for any $j \in J$.
Here we intentionally leave out the case of $j =~\emp$ in the definition of $\mu_j$ because the analysis will handle the dummy vertex separately.
The Poisson arrival model implies the following independence properties, which hold only asymptotically in online stochastic matching (see, e.g., Lemma 4 of \citet{JailletL:MOR:2014}).

\begin{lemma}\label{lem:independence-pair}
	In the Poisson arrival model, for any $j, k \in J^*$, online vertices of type $(*,j,k)$ follow a Poisson process with arrival rate $\mu_{jk}$, independent across different $(j,k)$ pairs.
\end{lemma}

\begin{proof}
    It holds because the online vertices follow a Poisson process and the probability that an online vertex samples $(j, k)$ is $\mu_{jk}$.
\end{proof}

As a corollary, we have a similar property for types $(*, j, *)$ for all $j \in J$.

\begin{lemma}
    \label{lem:independence-first}
    In the Poisson arrival model, independently for any $j \in J^*$, online vertices of type $(*,j,*)$ follow a Poisson process.
	The arrival rate is $\mu_j$ for any $j \in J$.
\end{lemma}

\subsection{Probability of Matching an Offline Vertex}

For an offline vertex type $j \in J$, $j$ may be matched in the following ways:
\begin{enumerate}
    \item $j$ is matched by an online vertex of type $(*,j,*)$;
    \item $j$ is matched by an online vertex of type $(*,\emp,j)$;
    \item Some $k\ne j$ is matched by type $(*,k,*)$ before the appearance of $(*,k,j)$.
\end{enumerate}
We remark that the above list is not exhaustive in general.
For example, three consecutive online vertices of type $(*,\ell,*), (*,\ell,k), (*,k,j)$ for some $k,\ell\neq j$ may match $\ell,k$ and finally $j$.

The probability that $j$ is matched by the first two cases is straightforward.
We next compute the probability that $j$ is matched by the last case.

\begin{lemma}
	\label{lem:prob-second}
	Consider any offline vertex $j \in J$.
	For any other offline vertex $k \in J \setminus \{j\}$, the probability that there is at least one online vertex of type $(*,k,j)$ after the first appearance of type $(*,k,*)$ is:
    \[
        \begin{cases}
            1 - \frac{\mu_k}{\mu_k - \mu_{kj}} e^{-\mu_{kj}} + \frac{\mu_{kj}}{\mu_k - \mu_{kj}} e^{-\mu_{k}} & \mu_k \ne \mu_{kj} ~; \\
            1 - e^{-\mu_k} - \mu_k e^{-\mu_k} & \mu_k = \mu_{kj} ~.
        \end{cases}
    \]
    Further, this is independent for different $k \in J\backslash\set{j}$, and is independent to online vertices of type $(*, j, *)$ and $(*, \emp, *)$.
\end{lemma}

\begin{proof}
    By Lemma~\ref{lem:independence-first}, the probability of having $\ell \ge 2$ online vertices of type $(*, k, *)$ is $\frac{\mu_k^\ell e^{-\mu_k}}{\ell!}$.
    For each of these online vertices, except the first one, its type is $(*, k, j)$ independently with probability $\frac{\mu_{kj}}{\mu_k}$ by Lemma~\ref{lem:independence-pair}.
    Therefore, the probability in the lemma is:
    \[
		\sum_{\ell = 2}^\infty \frac{\mu_{k}^\ell e^{-\mu_{k}}}{\ell!} \Big( 1 - \Big( 1 - \frac{\mu_{kj}}{\mu_k} \Big)^{\ell-1} \Big) = 
		\sum_{\ell = 2}^\infty \frac{\mu_{k}^\ell e^{-\mu_{k}}}{\ell!} 
        - 
        \sum_{\ell = 2}^\infty \frac{\mu_{k} (\mu_k - \mu_{kj})^{\ell-1} e^{-\mu_{k}}}{\ell!} 
        ~.
    \]

    By the Taylor series of $e^x$:
    \begin{align*}
        \sum_{\ell = 2}^\infty \frac{\mu_{k}^\ell e^{-\mu_{k}}}{\ell!}
        &
        = 1 - e^{-\mu_k} - \mu_k e^{-\mu_k} \\
        \sum_{\ell = 2}^\infty \frac{\mu_{k} (\mu_k - \mu_{kj})^{\ell-1} e^{-\mu_{k}}}{\ell!}
        &
        =
        \begin{cases}
            \frac{\mu_k e^{-\mu_{k}}}{\mu_k-\mu_{kj}} \left(e^{\mu_{k} - \mu_{kj}} - 1 - (\mu_{k} - \mu_{kj})\right) & \mu_k \ne \mu_{kj} ~; \\[1ex]
            0 & \mu_k = \mu_{kj} ~.
        \end{cases}
    \end{align*}

    Grouping terms by $e^{-\mu_k}$ and $e^{-\mu_{kj}}$ gives the probability in the lemma.
    Finally, the independence follows by Lemma~\ref{lem:independence-pair} and Lemma~\ref{lem:independence-first}.
\end{proof}

\paragraph{Auxiliary Function $\bm{\phi}$.}
Define $\phi(x,y)$ so that $e^{\phi(\mu_j, \mu_{jk})}$ equals the probability of having no online vertex of type $(*,j,k)$ after the first appearance of type $(*,j,*)$, including the case of having no vertex of type $(*,j,k)$.
In other words, with probability $e^{\phi(\mu_j, \mu_{jk})}$ the algorithm has never tried to match an online vertex of type $(*, j, *)$ to $k$.
By Lemma~\ref{lem:prob-second}:
\[
    \phi(x, y) \defeq \begin{cases}
        \ln \big( \frac{x}{x - y} e^{-y} - \frac{y}{x - y} e^{-x} \big) & x \ne y ~; \\[1ex]
        \ln (1+x) - x & x = y ~.
    \end{cases}
\]


%


\begin{lemma}
    \label{lem:prob-lower-bound}
    For any offline vertex $j \in J$, the meta algorithm matches it with probability at least:
    \[
        1 - e^{-\mu_j} \cdot e^{-\mu_{\emp j}} \cdot e^{\sum_{k \in J \setminus \set{j}} \phi(\mu_k,\mu_{kj})}
        ~.
    \]
\end{lemma}
\begin{proof}
    The algorithm does not match $j$ if and only if none of the following events happen:
    (1) there is no online vertex of type $(*, j, *)$;
    (2) there is no online vertex of type $(*, \emp, j)$; and
    (3) there is no online vertex of type $(*, k, j)$ after the first online vertex of type $(*, k, *)$, for some $k \in J \setminus \{ j \}$.
    There events are independent, and happen with probability $e^{-\mu_j}$, $e^{-\mu_{\emp j}}$, and $e^{\phi(\mu_k, \mu_{kj})}$ respectively by Lemma~\ref{lem:independence-pair}, Lemma~\ref{lem:independence-first}, and the definition of auxiliary function $\phi$.
\end{proof}

\subsection{Properties of the Auxiliary Function}

\begin{lemma}
	\label{lem:f-property}
    For any $x, y \in [0, 1]$, $\phi(x,y)$ is non-increasing and convex w.r.t.\ each coordinate.
\end{lemma}

\begin{proof}
    By symmetry, it suffices to prove it for $y$. 
    Equivalently, we need the following first-order and second-order partial derivatives in $y$ to be non-positive and non-negative respectively.
    %
    \begin{align*}
        \pdif{}{y} f(x,y) & = \frac{x \big( (1 - x + y) e^{-y} - e^{-x} \big)} {(x e^{-y} - y e^{-x}) (x - y)} ~,
        \\
        \pdifsqr{}{y}\phi(x,y) & = \frac{x \left( x e^{-2y} - y (x - y)^2 e^{-(x + y)} - (x - 2y) e^{-2x} \right)}{(x e^{-y} - y e^{-x})^2 (x - y)^2}
        ~.
    \end{align*}
    
    The first-order derivative is non-positive since $1 - x + y \le e^{-x+y}$.
    For the second-order derivative, consider the equation within the parentheses in the numerator.
    Since $x-y \in [-1, 1]$, it is at least: 
    \[
        x e^{-2y} - y e^{-(x + y)} - (x - 2y) e^{-2x}
        ~.
    \]

    Rearranging terms, we can write it as the sum of $x \big( e^{-y} - e^{-x} \big)^2$, $2 (x - y) e^{-x} \big(e^{-y} - e^{-x}\big)$, and $y e^{-(x+y)}$, all of which are nonnegative for any $x, y \ge 0$.
\end{proof}

As a corollary of its convexity, we obtain an upper bound of the value of $f$.

\begin{lemma}\label{lem:fxy-bound}
	For any $x,y\in[0,1]$, $\phi(x,y)\leq (\ln 2 - 1) xy$.
\end{lemma}

\begin{proof}
    By the convexity in Lemma~\ref{lem:f-property}, and that $\phi(x, 0) = \phi(0, y) = \phi(0, 0) =  0$:
    \[
        \phi(x, y) \le \phi(x, 1) \cdot y \le \phi(1, 1) \cdot xy
        ~.
    \]

    The lemma then follows by $\phi(1, 1) = \ln 2 - 1$.
\end{proof}

\subsection{Monotonicity in Unweighted Matching}

This subsection shows that the meta algorithm (Algorithm \ref{alg:pair-sampling-unweighted}) is monotone in the unweighted setting.
Hence, by Theorems \ref{thm:opt-compare} and \ref{thm:alg-compare}, the competitive ratios in Poisson arrival model also holds in online stochastic matching, up to a $1 - O(\Lambda^{-\frac{1}{2}})$ factor.

\begin{lemma}
    \label{lem:monotone-unweighted}
    For any distributions $D_i$'s, the meta algorithm is monotone in the unweighted case of online stochastic matching.
\end{lemma}

\begin{proof}
    The expected gain from the $\ell$-th online vertex equals the probability that at least one of its sampled offline vertices $j, k$ is still unmatched.
    Since the distribution of online types and the distributions $D_i$'s are time invariant, the above probability is non-increasing in $n$.
\end{proof}

\section{Unweighted Matching}
\label{sec:unweighted}

\subsection{Wasteful Correlated Sampling}
\label{correlated-pair-sampling}

Our starting point is the algorithm by \citet{JailletL:MOR:2014}, which we restate below.
We will refer to it as Wasteful Correlated Sampling because it may sample $k = j$ in some cases, and thus waste the second entry.
Further, we shall denote the sampling distributions as $D_i^{1}$ for $i \in I$ because they fall into a broader family of distributions $D_i^{c}$ for any $c \ge 1$, which we shall explain shortly in the next subsection.
Our final algorithm will be the limit case when $c = \infty$.

Let $\set{x_{ij}}_{i\in I,j\in J}$ be the optimal solution of \ref{eqn:natural-lp}. Define $x_{j}=\sum_{i \in i} x_{ij}$ for all $j\in J$, and $x_{i\emp} = \lambda_i - \sum_{j \in J} x_{ij}$ so that $\sum_{j \in J^*} x_{ij} = \lambda_i$.

\begin{definition}[Wasteful Correlated Sampling]
    \label{def:wasteful-correlated-sample}
    For any online type $i \in I$, a sample $(j, k)$ from $D_i^1$ is generated as follows:
    \begin{enumerate}[itemsep=0pt, topsep=5pt]
        \item
            Consider an interval $[0, \lambda_i)$.
            Align subintervals $I_j \subset [0, \lambda_i)$ of lengths $x_{ij}$ for $j \in J^*$ from left to right.
            See Figure~\ref{fig:correlated-sample}.
        \item
            Sample $\nu \in [0, \lambda_i)$ uniformly at random.
            Let $\nu' = \nu \pm \frac{\lambda_i}{2}$ such that $\nu' \in [0, \lambda_i)$.
            Note that $\nu$ and $\nu'$ are equally distributed.
        \item
            Let $j, k \in J^*$ be such that $\nu \in I_j$ and $\nu' \in I_k$.
        %
    \end{enumerate}
\end{definition}

\begin{figure*}
	\centering
	\begin{tikzpicture}
		\draw(0,0)--(2,0)[red,ultra thick];
		\draw(2,0)--(3.2,0)[blue,ultra thick];
		\draw(3.2,0)--(4,0)[green,ultra thick];
		\draw(4,0)--(5,0)[ultra thick];
		\draw(0,-0.1)--(0,0.1);
		\draw(2,-0.1)--(2,0.1);
		\draw(3.2,-0.1)--(3.2,0.1);
		\draw(4,-0.1)--(4,0.1);
		\draw(5,-0.1)--(5,0.1);
		\draw(1,0.3)node{$I_1$};
		\draw(2.6,0.3)node{$I_2$};
		\draw(3.6,0.3)node{$I_3$};
		\draw(4.5,0.3)node{$I_\emp$};
		\draw(1,-0.3)node{$0.4$};
		\draw(2.6,-0.3)node{$0.24$};
		\draw(3.6,-0.3)node{$0.16$};
		\draw(4.5,-0.3)node{$0.2$};
		\draw(-5,2)--(-7,0)[red,thick];
		\draw(-5,2)--(-5,0)[blue,thick];
		\draw(-5,2)--(-3,0)[green,thick];
		\filldraw[fill=white](-5,2)circle(0.3)node{$i$};
		\filldraw[fill=white](-7,0)circle(0.3)node{$1$};
		\filldraw[fill=white](-5,0)circle(0.3)node{$2$};
		\filldraw[fill=white](-3,0)circle(0.3)node{$3$};
		\draw(-4,2)node{$\lambda_i=1$};
		\draw(-7,1)node[red]{$x_{i1}=0.4$};
		\draw(-5,0.6)node[blue]{$x_{i2}=0.24$};
		\draw(-3,1)node[green]{$x_{i3}=0.16$};
	\end{tikzpicture}
	\caption{Illustration of intervals $I_1,I_2,I_3$ and $I_\emp$ for online type $i\in I$ with neighbors $1,2,3$.}
	\label{fig:correlated-sample}
\end{figure*}
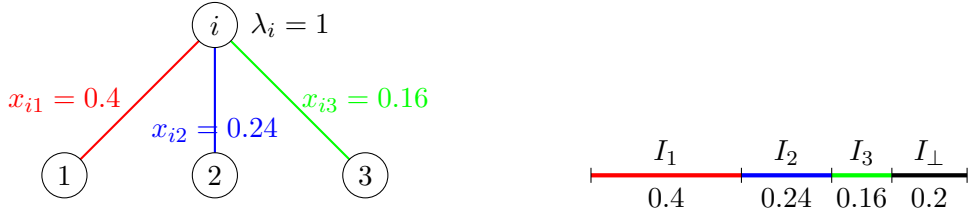



Let $\mu_{jk}(1) = \sum_{i \in I} \lambda_k D_i^1(j, k)$ for any $j, k \in J^*$, and $\mu_j(1) = \sum_{k \in J^*} \mu_{jk}(1)$ denote the arrival rates of online vertices of type $(*, j, k)$ and $(*, j, *)$ w.r.t.\ distributions $D_i^1$'s.

\begin{lemma}\label{lem:basic-sampling-property}
    Wasteful Correlated Sampling satisfies the following properties:
    \begin{enumerate}[itemsep=0pt, topsep=5pt]
        \item For any $j \in J$, $\mu_j(1) = x_j$.
        \item For any $j \in J$, $\sum_{k \in J^* \setminus \set{j}} \mu_{kj}(1)$ is at most $1$.
        \item For any $j \in J$, $\sum_{k \in J^* \setminus \set{j}} \mu_{jk}(1)$ is at least:
            \[
                \kappa(x_j) \defeq
                \begin{cases}
                    -\ln(1 - x_j) - x_j & 0 \le x_j \le \frac{1}{2} ~;
                    \\[1ex]
                    x_j - 1 + \ln 2 & \frac{1}{2} < x_j \le 1 ~.
                \end{cases}
            \]
        \item For any $j \ne k \in J^*$, $\mu_{jk}(1) = \mu_{kj}(1)$.
    \end{enumerate}
\end{lemma}

\begin{proof}
    \textbf{(Part 1)~}
    By definition, the probability that an online vertex of type $i$ samples $j$ as the first entry is $\frac{x_{ij}}{\lambda_i}$.
    Hence, we get that $\mu_j(1) = \sum_{i \in I} \lambda_i \cdot \frac{x_{ij}}{\lambda_i} = x_j$.
    \\[1ex]
    \textbf{(Part 2)~}
    By changing the order of summation, we get that:
    \[
        \sum_{k \in J^* \setminus \set{j}} \mu_{kj}(1) = \sum_{k \in J^* \setminus \set{j}} \sum_{i \in I} \lambda_i D_i^1(k, j) = \sum_{i \in I} \lambda_i \sum_{k \in J^* \setminus \set{j}}  D_i^1(k, j)
        ~.
    \]
    
    Further, an online vertex of type $i$ samples $j \in J$ as second entry with probability at most $\frac{x_{ij}}{\lambda_i}$;
    it may be smaller because the two entries may be equal in some cases.
    Hence:
    \[
        \sum_{k \in J^* \setminus \set{j}} \mu_{kj}(1) \le \sum_{i \in I} \lambda_i \frac{x_{ij}}{\lambda_i} = \sum_{i \in I} x_{ij} \leq 1
        ~.
    \]
    \textbf{(Part 3)~}
    The probability that an online vertex of type $i$ samples $j \in J$ as the first entry, and further samples a second entry $k \ne j$ equals $\min \{ \frac{x_{ij}}{\lambda_i}, 1 - \frac{x_{ij}}{\lambda_i} \}$, or equivalently, $\frac{x_{ij}}{\lambda_i} - \max \{ 2 \cdot \frac{x_{ij}}{\lambda_i} - 1, 0 \}$.
    Hence, letting $f(x) = \max \{ 2x-1, 0 \}$, we have:
    \[
        \sum_{k\in J^*\setminus\set{j}} \mu_{jk}(1) = \sum_{i \in I} \lambda_i \Big( \frac{x_{ij}}{\lambda_i} - f \Big( \frac{x_{ij}}{\lambda_i} \Big) \Big) = x_j - \sum_{i \in I} \lambda_i f \Big( \frac{x_{ij}}{\lambda_i} \Big)
        ~.
    \]
    
    Further by the converse of Jensen's inequality in Lemma~\ref{lem:converse-jensen-general}, this is at most:
    \[
        x_j-\int_0^{-\ln(1-x_j)} f \big( e^{-\lambda} \big) d \lambda
        =
        \begin{cases}
            - \ln(1 - x_j) - x_j & x_j \le \frac{1}{2} ~;
            \\[1ex]
            x_j - 1 +\ln 2 & x_j>\frac{1}{2} ~.
        \end{cases}
    \]
    \textbf{(Part 4)~}
    It follows by the symmetric joint distribution of $(\nu, \nu')$.
\end{proof}

We now present an analysis of the Pair Sampling algorithm with Wasteful Correlated Sampling that is simpler but weaker than the $0.706$ competitive raito by \citet{JailletL:MOR:2014}.
Nonetheless, we develop in the process some lemmas that are useful in the analysis of the final algorithm.

\begin{theorem}\label{thm:basic-pair-ratio}
    The competitive ratio of Pair Sampling with Wasteful Correlated Sampling in the unweighted case of online stochastic matching is at least:
    \[
        1 - \frac{1}{e} + \frac{1}{e} \big( 1 - \frac{2}{e} \big) \ln 2 > 0.699
        ~.
    \]
\end{theorem}

\begin{proof}
    We follow the framework of \citet{JailletL:MOR:2014}, except that the Poisson arrival model ensures true independence among online vertices of type $(*, j, *)$ for all $j \in J$, instead of the asymptotic independence in \citet{JailletL:MOR:2014}.
    Since $\mu_j(1) = x_j$ for any $j \in J$ due to Lemma~\ref{lem:basic-sampling-property}, by Lemma~\ref{lem:prob-lower-bound} the expected size of the algorithm's matching is at least:
    \[
        \alg \ge \sum_{j \in J} \Big( 1 - e^{-x_j} \cdot e^{-\mu_{\emp j}(1)} \cdot e^{\sum_{k\in J\setminus\set{j}} \phi(x_k, \mu_{kj}(1))} \Big)
        ~.
    \]

    Further by $\phi(x, y) \le (\ln 2 - 1) xy$ according to Lemma~\ref{lem:fxy-bound}, it is at least:
    \[
        \sum_{j \in J} \Big( 1 - e^{-x_j} \cdot e^{- \mu_{\emp j}(1) - (1 - \ln 2) \sum_{k\in J\setminus\set{j}} x_k \mu_{kj}(1)} \Big)
        ~.
    \]

    We artificially decrease $\mu_{\emp j}(1) = \mu_{j \emp}(1)$ (Lemma~\ref{lem:basic-sampling-property}) to $(1 - \ln 2) x_j \mu_{j \emp}(1)$ to mimic the form of the other terms as a preparation for the amortized argument.
    $\alg$ is then lower bounded by:
    \[
        \sum_{j \in J} \Big( 1 - e^{-x_j} \cdot e^{- (1 - \ln 2) ( x_j \mu_{j\emp}(1) +  \sum_{k\in J\setminus\set{j}} x_k \mu_{kj}(1) )} \Big)
        ~.
    \]

    Splitting each term as $1 - e^{-x_j} + e^{-x_j} \big( 1 - e^{- (1 - \ln 2) ( x_j \mu_{j\emp}(1) +  \sum_{k\in J\setminus\set{j}} x_k \mu_{kj}(1) )} \big)$, we get:
    \[
        \alg \ge \sum_{j \in J} \Big( \underbrace{\vphantom{\Big|} 1 - e^{-x_j}}_{\text{\normalsize (basic)}} + \underbrace{\frac{1}{e} \big( 1 - e^{- (1 - \ln 2) ( x_j \mu_{j\emp}(1) +  \sum_{k\in J\setminus\set{j}} x_k \mu_{kj}(1) )} \big)}_{\text{\normalsize (extra)}} \Big)
        ~.
    \]

    The key step is an amortized analysis that bounds the \emph{extra} part above.
    We state it as a lemma so that it can be used in the analysis the final algorithm.
    Informally, the amortization counts each vertex $j$'s \emph{basic} part, and its contribution to the \emph{extra} part of the other vertices.

\begin{lemma}
    \label{lem:amortize-extra}
    The \emph{extra} part is at least:
    \[
        \frac{1}{e} \big(1 - \frac{2}{e} \big) \sum_{j \in J} x_j \sum_{k \in J^* \setminus \set{j}} \mu_{jk}(1)
        ~.
    \]
\end{lemma}

\begin{proof}[Proof of Lemma~\ref{lem:amortize-extra}]
    By $1 - e^{-cx} \ge (1 - e^{-c}) x$ for any $c \ge 0$ and any $0 \le x \le 1$, this is at least:
    \[
        \frac{1}{e} \big(1 - \frac{2}{e} \big) \sum_{j \in J} \Big( x_j \mu_{j\emp}(1) + \sum_{k \in J \setminus \set{j}} x_k \mu_{kj}(1) \Big)
    \]

    Changing the order of summations in the second term proves the lemma.
\end{proof}

    By Lemma~\ref{lem:amortize-extra} and further by the third property of Lemma~\ref{lem:basic-sampling-property}, we have:
    \[
        \alg \ge \sum_{j \in J} \Big( 1 - e^{-x_j} + \frac{1}{e} \big(1-\frac{2}{e}\big) x_j \kappa(x_j) \Big)
        ~.
    \]

    For $x_j \le \frac{1}{2}$, the \emph{basic} part alone is sufficient because $1 - e^{-x} \ge (1-e^{-\frac{1}{2}}) 2x > 0.786 \cdot x$ for any $0 \le x \le \frac{1}{2}$.
    For $x_j > \frac{1}{2}$, we have $\kappa(x_j) = x_j - 1 + \ln 2$.
    We shall use the next lemma, whose proof is deferred to Appendix~\ref{app:proof-amortize-monotone} since it is simple but tedious calculus.
    
\begin{lemma}
    \label{lem:amortize-monotone}
    The function $\frac{1 - e^{-x}}{x} + \frac{1}{e} \big( 1 - \frac{2}{e} \big) x$ is decreasing in $x \in [\frac{1}{2}, 1]$.
\end{lemma}

    By Lemma~\ref{lem:amortize-monotone}, we have:
    \[
        1 - e^{-x_j} + \frac{1}{e} \big(1-\frac{2}{e}\big) x_j \big( x_j - 1 + \ln 2 \big) \ge \Big( 1 - \frac{1}{e} + \frac{1}{e} \big(1-\frac{2}{e}\big) \ln 2 \Big) x_j > 0.699 \cdot x_j
        ~.
    \]

    Hence, summing the inequalities for all offline vertices $j \in J$ proves the theorem.
\end{proof}

\subsection{Correlated Sampling}

Consider the wasteful case of $D_i^1$ in the previous subsection for some online vertex type $i \in I$, i.e., when there is some offline vertex $j^*$ such that $x_{ij^*} > \frac{1}{2} \lambda_i$.
In this case, Wasteful Correlated Sampling has a simpler and equivalent interpretation:
\begin{enumerate}[itemsep=0pt, topsep=5pt]
    \item Sample $j \in J$ with probability $\frac{x_{ij}}{\lambda_i}$.
    \item If $j \ne j^*$, let $k = j^*$.
    \item If $j = j^*$, sample $k \in J^* \setminus \set{j^*}$ with probability $\frac{x_{ik}}{x_{ij^*}}$, and $k = j^*$ with probability $\frac{\lambda_i - x_{ij^*}}{x_{ij^*}}$.
\end{enumerate}

This subsection considers a variant that is not wasteful by increasing the probability of sampling $k \in J^* \setminus \set{j^*}$ to $\frac{x_{ik}}{\lambda_i - x_{ij^*}}$ in the third step and, as a result, eliminating the case of $k = j^*$.
As intermediate steps in the analysis, we will more generally consider a family of $\beta$-Correlated Sampling algorithms for any $\beta \ge 1$.
Denote the corresponding distributions as $D_i^\beta$ for all online vertex types $i \in I$.
The unwasteful algorithm is the limit case when $\beta \to \infty$, for which case we omit $\beta$ and call it Correlated Sampling.

\begin{definition}[$\beta$-Correlated Sampling]
    \label{def:correlated-sample}
    For any online type $i \in I$, let $D_i^\beta = D_i^1$ if $x_{ij} \le \frac{1}{2} \lambda_i$ for all $j \in J^*$.
    Otherwise, a sample $(j, k)$ from $D_i^\beta$ is generated as follows:
    \begin{enumerate}[itemsep=0pt, topsep=5pt]
        \item Sample $j$ with probability $\frac{x_{ij}}{\lambda_i}$.
        \item If $j \ne j^*$, let $k = j^*$.
        \item If $j = j^*$, sample $k \in J^*$ with probability:
            \[
                \begin{cases}
                    \frac{x_{ik}}{x_{ij^*}} \cdot \min \Big\{ \beta, \frac{x_{ij*}}{\lambda_i - x_{ij^*}} \Big\}
                    &
                    k \ne j^*
                    ~; \\[2ex]
                    \max \Big\{ 1 - \frac{\beta (\lambda_i - x_{ij^*})}{x_{ij^*}} , 0 \Big\}
                    &
                    k = j^*
                    ~.
                \end{cases}
            \]
    \end{enumerate}
\end{definition}


Let $\mu_{jk}(\beta) = \sum_{i \in I} \lambda_k D_i^\beta(j, k)$ for any $j, k \in J^*$, and $\mu_j(\beta) = \sum_{k \in J^*} \mu_{jk}(\beta)$ denote the arrival rates of online vertices of type $(*, j, k)$ and $(*, j, *)$ w.r.t.\ distributions $D_i^\beta$'s.

\begin{lemma}\label{lem:advanced-sampling-property}
    For any $\beta \ge 1$, the following properties hold for $\beta$-correlated sampling.
    \begin{enumerate}
        \item For any $j \in J$, $\mu_j(\beta) = x_j$.
        \item For any $j \in J$, $\sum_{k \in J^* \setminus \set{j}} \mu_{kj}(\beta)$ is at most $\beta$.
        \item For any $j \in J$, $\sum_{k \in J^* \setminus \set{j}} \mu_{jk}(\beta)$ is at least:
            \[
                \kappa(\beta, x_j) \defeq
                \begin{cases}
                    - \beta \big( \ln(1 - x_j) + x_j \big) & 0 \le x_j \le \frac{1}{\beta+1} ~;
                    \\[1ex]
                    x_j - 1 + \beta \ln \frac{\beta+1}{\beta} & \frac{1}{\beta+1} < x_j \le 1 ~.
                \end{cases}
            \]
        \item For any $j \ne k \in J^*$, $\mu_{jk}(\beta) \le \beta \cdot \mu_{jk}(1)$.
        \item For any $j \ne k \in J^*$, $\mu_{jk}(\beta) \le \beta \cdot \mu_{kj}(\beta)$.
    \end{enumerate}
\end{lemma}

\begin{proof}
    \textbf{(Part 1)~}
    This is verbatim to the case of $\beta = 1$.
    By definition, the probability of an online vertex of type $i$ samples $j$ as the first entry is $\frac{x_{ij}}{\lambda_i}$.
    Hence, we get that $\mu_j(\beta) = \sum_{i \in I} \lambda_i \cdot \frac{x_{ij}}{\lambda_i} = x_j$.
    \\[1ex]
    \textbf{(Part 2)~}
    It follows by comparing the definitions, as the probabiltiy $D_i^\beta(j k) \le \beta D_i^1(j, k)$, then applying lemma \ref{lem:basic-sampling-property}.\\[1ex]
    \textbf{(Part 3)~}
    By changing the order of summation:
    \[
        \sum_{k \in J^* \setminus \set{j}} \mu_{jk}(\beta)
        =
        \sum_{k \in J^* \setminus \set{j}} \sum_{i \in I} \lambda_i D_i^\beta(j, k)
        =
        \sum_{i \in I} \lambda_i \sum_{k \in J^* \setminus \set{j}} D_i^\beta(j, k)
        ~.
    \]

    An online vertex of type $i$ samples $j \in J$ as first entry and further some $k \ne j$ as the second entry with probability $\min \big\{ \frac{x_{ij}}{\lambda_i}, \beta \big( 1 - \frac{x_{ij}}{\lambda_i} \big) \big\}$, or equivalently, $\frac{x_{ij}}{\lambda_i} - \max \{ (\beta+1) \frac{x_{ij}}{\lambda_i} - \beta, 0 \}$.
    Hence, letting $f(x) = \max \{ (\beta+1)x - \beta, 0 \}$:
    \[
        \sum_{k \in J^* \setminus \set{j}} \mu_{jk}(\beta) = \sum_{i \in I} \lambda_i \Big( \frac{x_{ij}}{\lambda_i} - f \Big( \frac{x_{ij}}{\lambda_i} \Big) \Big) = x_j - \sum_{i \in I} \lambda_i f \Big( \frac{x_{ij}}{\lambda_i} \Big)
        ~.
    \]
    
    Further by Lemma~\ref{lem:converse-jensen-general}, this is at most:
    \[
        x_j-\int_0^{-\ln(1-x_j)} f \big( e^{-\lambda} \big) d \lambda
        =
        \begin{cases}
            - \beta \big( \ln(1 - x_j) + x_j \big) & x_j \le \frac{1}{\beta+1} ~;
            \\[1ex]
            x_j - 1 + \beta \ln \frac{\beta+1}{\beta} & x_j>\frac{1}{\beta+1} ~.
        \end{cases}
    \]
    \textbf{(Part 4)~}
    It follows by $D_i^\beta(j k) \le \beta D_i^1(j, k)$.\\[1ex]
    \textbf{(Part 5)~}
    By the second part we have $\mu_{jk}(1) \le \mu_{jk}(\beta) \le \beta \cdot \mu_{jk}(1)$ for any $j \ne k \in J^*$.
    Hence, this part follows by $\mu_{jk}(1) = \mu_{kj}(1)$ due to Lemma~\ref{lem:basic-sampling-property}.
\end{proof}

With these properties of $\beta$-correlated sampling, we now prove our main theorem.

\begin{theorem}
    Pair Sampling with Correlated Sampling is at least \unweighted-competitive.
\end{theorem}

\begin{proof}
    Since $\mu_j(\infty) = x_j$ for any $j \in J$ due to Lemma~\ref{lem:advanced-sampling-property}, by Lemma~\ref{lem:prob-lower-bound} the expected size of the algorithm's matching is at least:
    \[
        \alg \geq \sum_{j \in J} \Big( 1 - e^{-x_j} \cdot e^{-\mu_{\emp j}(\infty)} \cdot e^{\sum_{k \in J \setminus\set{j}} \phi(x_k, \mu_{kj}(\infty))} \Big)
        ~.
    \]

    By the monotonicity of $\mu_{jk}(\beta)$'s in $\beta$, for $c = \frac{1}{1 - \ln 2}$, the above bound is at least:
    \[
        \sum_{j \in J} \Big( 1 - e^{-x_j} \cdot e^{-\mu_{\emp j}(c)} \cdot e^{\sum_{k \in J \setminus\set{j}} \phi(x_k, \mu_{kj}(c))} \Big)
        ~.
    \]

    By $\phi(x, y) \le (\ln 2 - 1) xy$ (Lemma~\ref{lem:fxy-bound}), by $\mu_{\emp j}(c) \ge \frac{1}{c} \mu_{j \emp}(c) = (1 - \ln 2) \mu_{j \emp}(c)$ (Lemma~\ref{lem:advanced-sampling-property}), and by $x_j \le 1$, this is at least:
    \[
        \alg \geq \sum_{j \in J} \Big( 1 - e^{-x_j} \cdot e^{-(1 - \ln 2) ( x_j \mu_{j\emp}(c) + \sum_{k \in J \setminus\set{j}} x_k \mu_{kj}(c))} \Big)
        ~.
    \]

    Splitting the term as $1 - e^{-x_j} + e^{-x_j} \big( 1 - e^{-(1 - \ln 2) ( x_j \mu_{j\emp}(c) + \sum_{k \in J \setminus\set{j}} x_k \mu_{kj}(c))} \big)$ for each offline vertex $j \in J$, by $x_j \le 1$ this is at least:
    \[
        1 - e^{-x_j} + \frac{1}{e} \big( 1 - e^{-(1 - \ln 2) ( x_j \mu_{j\emp}(c) + \sum_{k \in J \setminus\set{j}} x_k \mu_{kj}(c))} \big)
        ~.
    \]

    Further split the second part to get:
    \begin{align*}
        &
        \underbrace{\vphantom{\Big|} 1 - e^{-x_j}}_{\text{\normalsize (basic)}} + \underbrace{\frac{1}{e} \big( 1 - e^{-(1 - \ln 2) ( x_j \mu_{j\emp}(1) + \sum_{k \in J \setminus\set{j}} x_k \mu_{kj}(1))} \big)}_{\text{\normalsize (extra)}}
        \\
        & \quad
        + \underbrace{\frac{1}{e} \big( e^{-(1-\ln 2)( x_j\mu_{j\emp}(1) + \sum_{k\in J\setminus\set{j}} x_k\mu_{kj}(1))} - e^{-(1-\ln 2)( x_j\mu_{j\emp}(c) + \sum_{k\in J\setminus\set{j}} x_k\mu_{kj}(c))} \big)}_{\text{\normalsize (advanced)}}
        ~.
    \end{align*}

    \paragraph{Amortizing the Extra Part.}
    We will use the same amortized analysis in the previous subsection to bound the \emph{extra} part above.
    By Lemma~\ref{lem:amortize-extra}, the \emph{extra} part summing over $j$ is at least:
    \[
        \sum_{j \in J} \frac{1}{e} \big(1 - \frac{2}{e} \big) x_j \sum_{k \in J^* \setminus \set{j}} \mu_{jk}(1)
        ~.
    \]

    To simplify notation, for any $\beta \ge 1$ define:
    \[
        \mu_{j \to}(\beta) \defeq \sum_{k \in J^* \setminus \set{j}} \mu_{jk}(\beta)
        ~.
    \]

    Hence, we rewrite the bound as:
    \[
        \sum_{j \in J} \frac{1}{e} \big(1 - \frac{2}{e} \big) x_j \mu_{j \to}(1)
        ~.
    \]

    \paragraph{Amortizing the Advanced Part.}
    This part, omitting the $\frac{1}{e}$, can be written as:
    \begin{align*}
        &
        - \int_{\beta \in [1, c]} d e^{-(1-\ln 2)( x_j\mu_{j\emp}(\beta) + \sum_{k\in J\setminus\set{j}} x_k\mu_{kj}(\beta))}
        \\
        & \qquad
        =
        (1 - \ln 2) \int_{\beta \in [1, c]} e^{-(1-\ln 2)( x_j\mu_{j\emp}(\beta) + \sum_{k\in J\setminus\set{j}} x_k\mu_{kj}(\beta))} d \Big( x_j\mu_{j\emp}(\beta) + \sum_{k\in J\setminus\set{j}} x_k\mu_{kj}(\beta) \Big)
        ~.
    \end{align*}

    Next we bound the magnitude of the exponent.
    First, by $x_j, x_k \le 1$:
    \[
        x_j\mu_{j\emp}(\beta) + \sum_{k\in J\setminus\set{j}} x_k\mu_{kj}(\beta) \le \mu_{j\emp}(\beta) + \sum_{k\in J\setminus\set{j}} \mu_{kj}(\beta)
        ~.
    \]

    Further by $\mu_{jk}(\beta) \le \beta \mu_{jk}(1)$ and $\mu_{j\emp}(1) = \mu_{\emp j}(1)$, it is at most:
    \[
        \beta \Big( \mu_{\emp j}(\beta) + \sum_{k\in J\setminus\set{j}} \mu_{kj}(\beta) \Big) \le \beta
        ~.
    \]

    Hence, the advanced part is at least:
    \[
        (1-\ln 2)\int_{\beta \in [1, c]} e^{-(1-\ln 2)\beta} d \Big( x_j\mu_{j\emp}(\beta) + \sum_{k\in J\setminus\set{j}} x_k\mu_{kj}(\beta) \Big)
        ~.
    \]

    Summing over $j$ allows us to amortize as follows:
    \begin{align*}
        &
        (1-\ln 2)\int_{\beta \in [1, c]} e^{-(1-\ln 2)\beta} d \Big( \sum_{j\in J} x_j\mu_{j\emp}(\beta) + \sum_{j\in J} \sum_{k\in J\setminus\set{j}} x_k\mu_{kj}(\beta) \Big)
        ~.
        \\
        & \qquad
        = (1-\ln 2) \int_{\beta \in [1, c]} e^{-(1-\ln 2)\beta} d \sum_{j \in J} \sum_{k \in J^* \setminus\set{j}} x_j\mu_{jk}(\beta)
        \\
        & \qquad
        = \sum_{j \in J} (1-\ln 2) x_j \int_{\beta \in [1, c]} e^{-(1-\ln 2)\beta} d \mu_{j \to}(\beta)
        ~.
    \end{align*}

    Integrate by parts, for each $j$ the above equals:
    \[
        (1-\ln 2) x_j \Big(
            e^{-(1-\ln 2)c} \mu_{j \to}(c)
            - \frac{2}{e} \mu_{j \to}(1)
            + \int_{\beta \in [1, e]} \mu_{j \to}(\beta) e^{-(1-\ln 2) \beta} (1-\ln 2) d \beta
        \Big)
        ~.
    \]

    \paragraph{Putting Everything Together.}
    The sum of the lower bounds above for the \emph{extra} and \emph{advanced} parts is:
    \begin{align*}
        &
        \sum_{j \in J} \frac{x_j}{e} \Big( 
            \big( 1 - \ln 2 \big) e^{-(1-\ln 2)c} \mu_{j \to}(c)
            + \Big( 1 - \frac{2}{e} - \frac{2(1-\ln 2)}{e} \Big) \mu_{j \to}(1)
        \\
        & \qquad\qquad
            + \big( 1 - \ln 2 \big) \int_{\beta \in [1, e]} \mu_{j \to}(\beta) e^{-(1-\ln 2) \beta} (1-\ln 2) d \beta
        \Big)
        ~.
    \end{align*}

    Since the coefficients of $\mu_{j \to}(\beta)$ are positive for all $\beta \in [1, c]$, the minimum is achieved when $\mu_{j \to}(\beta) = \kappa(\beta, x_j)$ subject to the fourth property of Lemma~\ref{lem:advanced-sampling-property}.
    Hence we conclude that:
    \[
        \alg \ge \sum_{j \in J} \Big( 1 - e^{-x_j} + \frac{1}{e} \big( 1-\frac{2}{e} \big) x_j \kappa(1, x_j) + \frac{1-\ln 2}{e} x_j \int_{\beta \in [1, c]} e^{-(1-\ln 2)\beta} d \kappa(\beta, x_j) \Big)
        ~.
    \]

    For $x_j \le \frac{1}{2}$, the \emph{basic} part alone is sufficient because $1 - e^{-x} \ge (1-e^{-\frac{1}{2}}) 2x > 0.786 \cdot x$ for any $0 \le x \le \frac{1}{2}$.
    For $x_j > \frac{1}{2}$, we have $\kappa(\beta, x_j) = x_j - 1 + \beta \ln \frac{\beta+1}{\beta}$.
    Hence $j$'s contribution equals: 
    \[
        1 - e^{-x_j} + \frac{1}{e} \big( 1 - \frac{2}{e} \big) x_j \big(x_j - 1 + \ln 2 \big) + \frac{1 - \ln 2}{e} x_j \int_1^c e^{-(1-\ln 2) \beta} \big( \ln \tfrac{\beta+1}{\beta} - \tfrac{1}{\beta+1} \big) d\beta
        ~.
    \]

    By Lemma~\ref{lem:amortize-monotone}, the first three terms sum to at least $\Big( 1 - \frac{1}{e} + \frac{1}{e} \big(1-\frac{2}{e}\big) \ln 2 \Big) x_j$.
    The last integral does not seem to admit a closed-form solution so we calculate it numerically, and the above value is greater than $\unweighted x_j$.
    Hence, summing over all offline vertices $j \in J$ proves the theorem.
    %
\end{proof}

\section{Vertex-weighted Matching}
\label{sec:vertex-weighted}

This section considers the vertex-weighted problem.
Each offline vertex $j \in J$ has a non-negative weight $w_j$.
The objective is to maximize the sum of weights of the matched offline vertices.

\subsection{Failure of Correlated Sampling in Vertex-weighted Matching}

Recall the amortized analysis of correlated sampling in the last section.
It divides the probability that an offline vertex is matched into two parts, \emph{basic} and \emph{extra}. 
Then, it proves for any offline vertex that the sum of its \emph{basic} part and its contribution to the \emph{extra} parts of the other vertices is at least its contribution to the LP objective times the competitive ratio.
In the presence of vertex weights, however, the contribution of an offline vertex to the \emph{extra} parts of the other vertices are scaled by their weights, which could be negligible compared to its own weight.

Why do we need amortization to begin with?
Recall the probability that an offline vertex is matched given by Lemma~\ref{lem:prob-lower-bound}:
\[
    1 - e^{-\mu_j} \cdot e^{-\mu_{\emp j}} \cdot e^{\sum_{k \in J \setminus \set{j}} \phi(\mu_k,\mu_{kj})}
    ~.
\]
Although we can lower bound the total resampling mass $\sum_{k \in J \setminus \{j\}} \mu_{kj}$, their contribution to the above equation could be negligible if they come from many vertices $k$ whose $\mu_k$ are close to $0$.
It would be great if we could replace $e^{\phi(\mu_k, \mu_{kj})}$ in the above equation with $e^{-\beta \cdot \mu_{kj}}$ for some constant $\beta > 0$ by modifying the sampling distributions $D_i$'s appropriately.

\subsection{Amortized Correlated Sampling}
\label{sec:amortize-correlated-sample}

This subsection demonstrates how to obtain the above property for $\beta = 0.299$.
The key observation is that the problematic vertices $k$ with tiny $\mu_k$ satisfy:
\[
    1 - e^{-\mu_k} \approx \mu_k 
    ~.
\]

Therefore, the probability that such a vertex is matched by the algorithm is well above its contribution to the LP times the competitive ratio.
We could afford to drop it in the first option with some probability even if it is not yet matched, and to directly consider the second option.


\begin{definition}[Amortized Correlated Sampling]
    \label{def:amortize-correlated-sample}
    For any online type $i \in I$, define $D_i$ as:
    \begin{enumerate}[itemsep=0pt, topsep=5pt]
        \item Sample $(j,k)$ from $D_i^1$ as defined in Subsection~\ref{correlated-pair-sampling}. 
        \item With probability $\delta(x_j)=\max\set{\frac{\beta - (1-\ln 2) x_j}{1-2(1-\ln 2)x_j},0}$, replaces $j$ with $\emp$.
    \end{enumerate}
\end{definition}

Recall that $\mu_j(1)$'s and $\mu_{jk}(1)$'s are the probability given by $D_i^1$'s in Wasteful Correlated Sampling, i.e., without the second step above that replaces $j$ with $\emp$ with certain probability.
We establish below the properties of $\mu_j$'s and $\mu_{jk}$'s from Amortized Correlated Sampling in relation to their counterparts in Wasteful Correlated Sampling.

\begin{lemma}\label{lem:weighted-sampling-property}
    Amortized Correlated Sampling satisfies the following properties:
	\begin{enumerate}
        \item For any $j \in J$, $\mu_j = (1-\delta(x_j)) \mu_j(1)$;
        \item For any $j \ne k \in J$, $\mu_{jk} = (1-\delta(x_j))\mu_{jk}(1)$;	
        \item For any $k \in J$, $\mu_{\emp k} = \mu_{\emp k}(1) + \sum_{j \ne k} \delta(x_j) \mu_{jk}(1)$;
		
	\end{enumerate}
\end{lemma}

We can now state the modified version of Lemma~\ref{lem:prob-lower-bound}.

\begin{lemma}
    \label{lem:prob-amortized}
    For any offline vertex $j \in J$, Pair Sampling with Amortized Correlated Sampling matches $j$ with probability at least:
    \[
        1 - e^{-(1 - \delta(x_j)) \mu_j(1)} \cdot e^{-\mu_{\emp j}(1)} \cdot e^{- \beta \sum_{k\in J \setminus \set{j}} \mu_{kj}(1)}
        ~.
    \]
\end{lemma}

\begin{proof}
    By Lemma~\ref{lem:prob-lower-bound} and Lemma~\ref{lem:weighted-sampling-property}, the probability in the lemma equals:
    \[
        1 - e^{-(1-\delta(x_j)) \mu_j(1)} \cdot e^{-\mu_{\emp j}(1) - \sum_{k \in J \setminus \set{j}} \delta(x_k) \mu_{kj}(1) } \cdot e^{\sum_{k \in J \setminus \set{j}} \phi((1-\delta(x_k))\mu_k, (1-\delta(x_k))\mu_{kj})}
    \]

    Comparing to the equation in the lemma, it remains to show that for any $k \ne J \setminus \set{j}$:
    \[
        - \delta(x_k) \mu_{kj}(1) + f \big( (1-\delta(x_k)) \mu_k(1), (1-\delta(x_k)) \mu_{kj}(1) \big) \le - \beta \mu_{kj}(1)
        ~.
    \]

    By Lemma \ref{lem:fxy-bound}, the second term above is bounded by:
    \begin{align*}
        f \big( (1-\delta(x_k))\mu_k(1), (1-\delta(x_k)) \mu_{kj}(1) \big)
        &
        \leq -(1-\ln 2)(1-\delta(x_k))^2 \mu_k(1) \mu_{kj}(1)
        \\
        &
        \leq -(1-\ln 2)(1-2\delta(x_k)) \mu_k(1) \mu_{kj}(1)
        ~.
    \end{align*}

    Hence, it reduces to:
    \[
        \delta(x_k) + (1-\ln 2) (1-2\delta(x_k)) \mu_k(1) \ge \beta
        ~.
    \]

    Recall that $\mu_k(1) = x_k$.
    The choice of $\delta(x)$ ensures the inequality $(1-\ln 2)(1-2\delta(x))x+\delta(x) \geq \beta$ for any $0 \leq x \leq 1$.
\end{proof}

\begin{theorem}
    \label{thm:vertex-weighted}
	Pair Sampling with Amortized Correlated Sampling is at least 0.7009-competitive.
\end{theorem}

\begin{proof}
    It suffices to show that for any offline vertex $j$, the algorithm matches it with probability at least $0.7009 x_j$.
    By Lemma~\ref{lem:prob-amortized}, $j$ is matched with probability:
    \[
        1 - e^{-(1 - \delta(x_j)) \mu_j(1)} \cdot e^{-\mu_{\emp j}(1)} \cdot e^{- \beta \sum_{k\in J \setminus \set{j}} \mu_{kj}(1)}
        \ge 
        1 - e^{-(1 - \delta(x_j)) \mu_j(1)} \cdot e^{- \beta \sum_{k \in J^* \setminus \set{j}} \mu_{kj}(1)}
        ~.
    \]

    By Lemma~\ref{lem:basic-sampling-property}, we have $\mu_j(1) = x_j$ and $\sum_{k \in J^* \setminus \set{j}} \mu_{kj}(1) \ge \kappa(x_j)$.
    Hence, this is at least:
    \[
		1 - e^{-(1-\delta(x_j)) x_j - \beta \kappa(x_j)}
        ~.
    \]
	
    We numerically verified that $1-e^{-(1-\delta(x))x-\beta \kappa(x)}\geq 0.7009x$ for any $0\leq x\leq 1$.
\end{proof}

\subsection{Approximate Monotonicity}

Unlike the unweighted case, we can only prove an approximate monotonicity of the vertex-weighted matching algorithm. By Theorems \ref{thm:opt-compare} and \ref{thm:alg-compare}, the competitive ratio in the Poisson arrival model also holds in online stochastic matching, up to a $1 - O(\alpha\Lambda^{-\frac{1}{2}})$ factor for a constant $\alpha$.

\begin{lemma}
	\label{lem:monotone-weighted}
	Pair Sampling with Amortized Correlated Sampling is $O(1)$-approximately monotone in the vertex-weighted case of online stochastic matching.
\end{lemma}

\begin{proof}
	For any offline vertex $j\in J$, let $P_j(n)$ be the probability that $j$ is unmatched at the arrival of $n$-th online vertex.
	By definition $P_j(n)$ is non-increasing over $n$.
	
	Let $a_j$ be the probability that $j$ is at least one of the two options in one round.
	Let $b_j$ be the probability that $j$ is the first choice in one round.
	Since the distribution of online types and the sampling distributions $D_i$'s are time invariant, $a_j, b_j$ are constant throughout the process. Hence the probability that $j$ is matched exactly by $n$-th online vertex is upper bounded by $a_j P_j(n)$, and lower bounded by $b_j P_j(n)$.
	Recall that for any $n \ge 1$, $\alg_n$ denote the expected weight that the algorithm gets from matching the $n$-th online vertex.
	We have:
	\begin{equation*}
		\sum_{j\in J}w_j b_j P_j(n) \leq \alg_n \leq \sum_{j \in J} w_j a_j P_j(n).
	\end{equation*}
	
	Further using the monotonicity of $P_j(n)$, for any $\ell < n$:
	\begin{equation*}
		\frac{\alg_n}{\alg\ell} \leq \frac{\sum_{j\in J}w_ja_jP_j(n)}{\sum_{j\in J}w_jb_jP_j(\ell)} \leq \max_{j\in J}\frac{w_ja_jP_j(n)}{w_jb_jP_j(\ell)} \leq \max_{j\in J}\frac{a_j}{b_j}.
	\end{equation*}
    To give an upper bound of $\frac{a_j}{b_j}$, note that in $D_i^1$ defined in Subsection~\ref{correlated-pair-sampling}, the two options are equally distributed. In amortized correlated sampling, the first choice is dropped with probability $\delta(x_j)$. Therefore, $\frac{a_j}{b_j}\leq \frac{2}{1-\delta(x_j)}\leq\frac{2}{1-\beta}$,
	which is a constant (recall that $\beta = 0.299$).
	Therefore, the algorithm is $O(1)$-approximately monotone.
\end{proof}

\bibliographystyle{plainnat}
\bibliography{matching}

\begin{thebibliography}{26}
\providecommand{\natexlab}[1]{#1}
\providecommand{\url}[1]{\texttt{#1}}
\expandafter\ifx\csname urlstyle\endcsname\relax
  \providecommand{\doi}[1]{doi: #1}\else
  \providecommand{\doi}{doi: \begingroup \urlstyle{rm}\Url}\fi

\bibitem[Aggarwal et~al.(2011)Aggarwal, Goel, Karande, and
  Mehta]{AggarwalGKM:SODA:2011}
Gagan Aggarwal, Gagan Goel, Chinmay Karande, and Aranyak Mehta.
\newblock Online vertex-weighted bipartite matching and single-bid budgeted
  allocations.
\newblock In \emph{Proceedings of the 22nd Annual ACM-SIAM Symposium on
  Discrete Algorithms}, pages 1253--1264. SIAM, 2011.

\bibitem[Bahmani and Kapralov(2010)]{BahmaniK:ESA:2010}
Bahman Bahmani and Michael Kapralov.
\newblock Improved bounds for online stochastic matching.
\newblock In \emph{Proceedings of the 18th Annual European Symposium on
  Algorithms}, pages 170--181, 2010.

\bibitem[Brubach et~al.(2020)Brubach, Sankararaman, Srinivasan, and
  Xu]{Brubach:Algorithmica:2020}
Brian Brubach, Karthik~Abinav Sankararaman, Aravind Srinivasan, and Pan Xu.
\newblock Online stochastic matching: New algorithms and bounds.
\newblock \emph{Algorithmica}, page 2737–2783, 2020.

\bibitem[Buchbinder et~al.(2007)Buchbinder, Jain, and
  Naor]{BuchbinderJN:ESA:2007}
Niv Buchbinder, Kamal Jain, and Joseph~Seffi Naor.
\newblock Online primal-dual algorithms for maximizing ad-auctions revenue.
\newblock In \emph{European Symposium on Algorithms}, pages 253--264. Springer,
  2007.

\bibitem[Devanur and Hayes(2009)]{DevanurH:EC:2009}
Nikhil~R Devanur and Thomas~P Hayes.
\newblock The {AdWords} problem: online keyword matching with budgeted bidders
  under random permutations.
\newblock In \emph{Proceedings of the 10th ACM conference on Electronic
  commerce}, pages 71--78, 2009.

\bibitem[Devanur et~al.(2013)Devanur, Jain, and Kleinberg]{DevanurJK:SODA:2013}
Nikhil~R Devanur, Kamal Jain, and Robert~D Kleinberg.
\newblock Randomized primal-dual analysis of ranking for online bipartite
  matching.
\newblock In \emph{Proceedings of the 24th Annual ACM-SIAM Symposium on
  Discrete Algorithms}, pages 101--107. SIAM, 2013.

\bibitem[Devanur et~al.(2016)Devanur, Huang, Korula, Mirrokni, and
  Yan]{DevanurHKMY:TEAC:2016}
Nikhil~R Devanur, Zhiyi Huang, Nitish Korula, Vahab~S Mirrokni, and Qiqi Yan.
\newblock Whole-page optimization and submodular welfare maximization with
  online bidders.
\newblock \emph{ACM Transactions on Economics and Computation (TEAC)},
  4\penalty0 (3):\penalty0 1--20, 2016.

\bibitem[Fahrbach et~al.(2020)Fahrbach, Huang, Tao, and
  Zadimoghaddam]{FahrbachHTZ:FOCS:2020}
Matthew Fahrbach, Zhiyi Huang, Runzhou Tao, and Morteza Zadimoghaddam.
\newblock Edge-weighted online bipartite matching.
\newblock In \emph{Proceedings of the 61st Annual IEEE Symposium on Foundations
  of Computer Science}, 2020.

\bibitem[Feldman et~al.(2009{\natexlab{a}})Feldman, Korula, Mirrokni,
  Muthukrishnan, and Pál]{FeldmanKMMP:WINE:2009}
Jon Feldman, Nitish Korula, Vahab Mirrokni, Shanmugavelayutham Muthukrishnan,
  and Martin Pál.
\newblock Online ad assignment with free disposal.
\newblock In \emph{Proceedings of the 5th International Workshop on Internet
  and Network Economics}, pages 374--385. Springer, 2009{\natexlab{a}}.

\bibitem[Feldman et~al.(2009{\natexlab{b}})Feldman, Mehta, Mirrokni, and
  Muthukrishnan]{FeldmanMMM:FOCS:2009}
Jon Feldman, Aranyak Mehta, Vahab Mirrokni, and S.~Muthukrishnan.
\newblock Online stochastic matching: Beating $1-\frac{1}{e}$.
\newblock In \emph{Proceedings of 50th Annual Symposium on Foundations of
  Computer Science}, pages 117--126, 2009{\natexlab{b}}.

\bibitem[Goel and Mehta(2008)]{GoelM:SODA:2008}
Gagan Goel and Aranyak Mehta.
\newblock Online budgeted matching in random input models with applications to
  adwords.
\newblock In \emph{SODA}, volume~8, pages 982--991, 2008.

\bibitem[Goyal and Udwani(2020)]{GoyalU:EC:2020}
Vineet Goyal and Rajan Udwani.
\newblock Online matching with stochastic rewards: Optimal competitive ratio
  via path based formulation.
\newblock In \emph{Proceedings of the 21st ACM Conference on Economics and
  Computation}, 2020.

\bibitem[Haeupler et~al.(2011)Haeupler, Mirrokni, and
  Zadimoghaddam]{HaeuplerMZ:WINE:2011}
Bernhard Haeupler, Vahab~S. Mirrokni, and Morteza Zadimoghaddam.
\newblock Online stochastic weighted matching: Improved approximation
  algorithms.
\newblock In \emph{Proceedings of the 7th International Conference on Internet
  and Network Economics}, pages 170--181, 2011.

\bibitem[Huang and Zhang(2020)]{HuangZ:STOC:2020}
Zhiyi Huang and Qiankun Zhang.
\newblock Online primal dual meets online matching with stochastic rewards:
  configuration lp to the rescue.
\newblock In \emph{Proceedings of the 52nd Annual ACM SIGACT Symposium on
  Theory of Computing}, pages 1153--1164, 2020.

\bibitem[Huang et~al.(2019)Huang, Tang, Wu, and Zhang]{HuangTWZ:TALG:2019}
Zhiyi Huang, Zhihao~Gavin Tang, Xiaowei Wu, and Yuhao Zhang.
\newblock Online vertex-weighted bipartite matching: Beating $1-\frac{1}{e}$
  with random arrivals.
\newblock \emph{ACM Transactions on Algorithms}, 15\penalty0 (3):\penalty0
  1--15, 2019.

\bibitem[Huang et~al.(2020)Huang, Zhang, and Zhang]{HuangZZ:FOCS:2020}
Zhiyi Huang, Qiankun Zhang, and Yuhao Zhang.
\newblock Adwords in a panorama.
\newblock In \emph{Proceedings of the 61st Annual IEEE Symposium on Foundations
  of Computer Science}, 2020.

\bibitem[Jaillet and Lu(2014)]{JailletL:MOR:2014}
Patrick Jaillet and Xin Lu.
\newblock Online stochastic matching: New algorithms with better bounds.
\newblock \emph{Mathematics of Operations Research}, 39\penalty0 (3):\penalty0
  624--646, 2014.

\bibitem[Jin and Williamson(2020)]{JinW:arXiv:2020}
Billy Jin and David~P Williamson.
\newblock Improved analysis of ranking for online vertex-weighted bipartite
  matching.
\newblock \emph{arXiv preprint arXiv:2007.12823}, 2020.

\bibitem[Karande et~al.(2011)Karande, Mehta, and Tripathi]{KarandeMT:STOC:2011}
Chinmay Karande, Aranyak Mehta, and Pushkar Tripathi.
\newblock Online bipartite matching with unknown distributions.
\newblock In \emph{Proceedings of the 43rd Annual ACM Symposium on Theory of
  Computing}, pages 587--596, 2011.

\bibitem[Karp et~al.(1990)Karp, Vazirani, and Vazirani]{KarpVV:STOC:1990}
Richard~M Karp, Umesh~V Vazirani, and Vijay~V Vazirani.
\newblock An optimal algorithm for on-line bipartite matching.
\newblock In \emph{Proceedings of the 22nd Annual ACM Symposium on Theory of
  Computing}, pages 352--358, 1990.

\bibitem[Mahdian and Yan(2011)]{MahdianY:STOC:2011}
Mohammad Mahdian and Qiqi Yan.
\newblock Online bipartite matching with random arrivals: an approach based on
  strongly factor-revealing {LP}s.
\newblock In \emph{Proceedings of the 43rd Annual ACM Symposium on Theory of
  Computing}, pages 597--606, 2011.

\bibitem[Manshadi et~al.(2012)Manshadi, Oveis~Gharan, and
  Saberi]{ManshadiOS:MOR:2012}
Vahideh~H Manshadi, Shayan Oveis~Gharan, and Amin Saberi.
\newblock Online stochastic matching: Online actions based on offline
  statistics.
\newblock \emph{Mathematics of Operations Research}, 37\penalty0 (4):\penalty0
  559--573, 2012.

\bibitem[Mehta(2013)]{Mehta:FTTCS:2013}
Aranyak Mehta.
\newblock Online matching and ad allocation.
\newblock \emph{Foundations and Trends in Theoretical Computer Science},
  8\penalty0 (4):\penalty0 265--368, 2013.

\bibitem[Mehta and Panigrahi(2012)]{MehtaP:FOCS:2012}
Aranyak Mehta and Debmalya Panigrahi.
\newblock Online matching with stochastic rewards.
\newblock In \emph{Proceedings of the 53rd Annual IEEE Symposium on Foundations
  of Computer Science}, pages 728--737. IEEE, 2012.

\bibitem[Mehta et~al.(2007)Mehta, Saberi, Vazirani, and
  Vazirani]{MehtaSVV:JACM:2007}
Aranyak Mehta, Amin Saberi, Umesh Vazirani, and Vijay Vazirani.
\newblock Adwords and generalized online matching.
\newblock \emph{Journal of the ACM}, 54\penalty0 (5):\penalty0 22--es, 2007.

\bibitem[Mehta et~al.(2014)Mehta, Waggoner, and
  Zadimoghaddam]{MehtaWZ:SODA:2014}
Aranyak Mehta, Bo~Waggoner, and Morteza Zadimoghaddam.
\newblock Online stochastic matching with unequal probabilities.
\newblock In \emph{Proceedings of the 26th Annual ACM-SIAM Symposium on
  Discrete Algorithms}, pages 1388--1404. SIAM, 2014.

\end{thebibliography}

\appendix
\section{Comparisons with Existing Linear Programs}
\label{app:lp}

This section compares the natural LP with those in the previous works, and show that the previosu LP are all relaxations of the natural LP. We first restate the natural LP.
\begin{equation}
	\tag{\natlp}
	\begin{aligned}
		\text{maximize} \quad & \sum_{i \in I} \sum_{j \in J} w_{ij} x_{ij} \\
		\text{subject to} \quad &
		\sum_{j \in J} x_{ij} \le \lambda_i && \forall i \in I \\
		&
		\sum_{i \in S} x_{ij} \le 1 - \exp \Big( - \sum_{i \in S} \lambda_i \Big) && \forall j \in J, \forall S \subseteq I \\
		& x_{ij} \ge 0 && \forall i \in I, \forall j \in J
	\end{aligned}
\end{equation} 

\subsection{General Arrival Rates}

\paragraph{Jaillet-Lu Linear Program.}
\citet{JailletL:MOR:2014} considered the following LP for the unweighted matching with general arrival rates.
Let $\jllp$ denote its optimal value.
\begin{equation}
	\label{eqn:jaillet-lu-lp}
	\tag{\jllp}
	\begin{aligned}
		\text{maximize} \quad & \sum_{(i, j) \in E} x_{ij} \\
		\text{subject to} \quad
		& \sum_{j \in J} x_{ij} \le \lambda_i && \forall i \in I \\
        & \sum_{i \in I} x_{ij} \le 1 && \forall j \in J \\
		& \sum_{i \in I} (2x_{ij} - \lambda_i)^+ \le 1 - \ln 2 && \forall j \in J \\
		& x_{ij} \ge 0 && \forall i \in I, \forall j \in J
	\end{aligned}
\end{equation}

The main difference between this Jaillet-Lu LP and ours is the third constraint, due to \citet*{ManshadiOS:MOR:2012}. We now show that the constraints of the natural LP imply this constraint, therefore the natural LP is a better upper bound of \opt.

\begin{lemma}
	\label{lem:jaillet-lu-lp}
	$\natlp \le \jllp$.
\end{lemma}

\begin{proof}
    We will prove that any feasible solution of the natural LP is also feasible for the Jaillet-Lu LP.
    The first constraint is in both LPs.
    The second constraint holds in the natural LP since $\sum_{i\in I}x_{ij}\leq 1-\exp\left(\sum_{i \in I}\lambda_i\right)\leq 1.$ 
    The third constraint follows by the converse of Jensen's inequality in Lemma~\ref{lem:converse-jensen-general}.
    Let $f(x) = \max \{2x-1, 0\}$:
    \[
        \sum_{i \in I} (2 x_{ij} - \lambda_i)^+ = \sum_{i \in I} \lambda_i f \Big( \frac{x_{ij}}{\lambda_i} \Big)
        \le
        \int_0^\infty f(e^{-\lambda}) d\lambda
        =
        1 - \ln 2
        ~.
    \]
\end{proof}

\subsection{Integral Arrival Rates}

\paragraph{Jaillet-Lu Linear Program.}
For the special case of integral arrival rates, i.e., when $\lambda_i = 1$ for all online types $i \in I$, \citet{JailletL:MOR:2014} considered a different LP.
\begin{align*}
    \text{maximize} \quad & \sum_{(i, j) \in E} x_{ij} \\
    \text{subject to} \quad
    & \sum_{j \in J} x_{ij} \le 1 && \forall i \in I \\
    & \sum_{i \in I} x_{ij} \le 1 && \forall j \in J \\
    & 0 \le x_{ij} \le \frac{2}{3} && \forall i \in I, \forall j \in J
\end{align*}

This is a relaxation of the natural LP, keeping only a subset of the second constraint, either at the limit with $S = I$ and $1$ on the right-hand-side, or when $S = \{ i \}$ is a singleton set, i.e.:
\[
    \sum_{i \in I} x_{ij} \le 1 - \frac{1}{e}
    ~.
\]

It further relaxes the $1 - \frac{1}{e}$ to $\frac{2}{3}$.

\paragraph{Brubach-Sankararaman-Srinivasan-Xu Linear Program.}
The LP employed by \citet*{Brubach:Algorithmica:2020} in the special case of integral arrival rates is the closest to ours.
It is a relaxation of the natural LP, keeping only a subset of the second constraint, either at the limit with $S = I$ and $1$ on the right-hand-side, or when the subset $S$ consists of only one or two online types.
\begin{align*}
    \text{maximize} \quad & \sum_{(i, j) \in E} w_i x_{ij} \\
    \text{subject to} \quad & \sum_{j \in N_i} x_{ij} \le 1 && \forall i \in I \\
    & \sum_{i \in N_j} x_{ij} \le 1 && \forall j \in J \\
    &0\leq x_{ij} \leq 1-e^{-1} && \forall i \in I, j \in J\\[2ex]
    &x_{i_1j}+x_{i_2j} \leq 1-e^{-2} && \forall i_1\neq i_2 \in I, j \in J
\end{align*}

\section{Omitted Proofs}

\subsection{Poisson Tail Bound}
\label{app:proof-poisson-tail}
\begin{lemma}
	\begin{equation*}
		\sum_{m=\Lambda+1}^\infty \frac{\Lambda^m e^{-\Lambda}}{m!} \frac{m-\Lambda}{\Lambda} = O(\Lambda^{-\frac{1}{2}}).
	\end{equation*}
	\begin{proof}
        We shall simplify the left-hand-side as follows:
		\begin{align*}
			\sum_{m=\Lambda+1}^\infty \frac{\Lambda^m e^{-\Lambda}}{m!} \frac{m-\Lambda}{\Lambda}
			= & \sum_{m=\Lambda+1}^\infty \frac{\Lambda^{m-1} e^{-\Lambda}}{m!} (m-\Lambda)
			~\\
			= & \sum_{m=\Lambda+1}^\infty \frac{\Lambda^{m-1} e^{-\Lambda}}{(m-1)!} - \sum_{m=\Lambda+1}^\infty \frac{\Lambda^m e^{-\Lambda}}{m!}
			~\\
			= & ~ \frac{\Lambda^\Lambda e^{-\Lambda}}{\Lambda!}
            ~.
		\end{align*}

        It then follows by Stirling's formula.
	\end{proof}
\end{lemma}

\subsection{Proof of Lemma \ref{lem:amortize-monotone}}
\label{app:proof-amortize-monotone}

\begin{proof}
    To prove that $\frac{1 - e^{-x}}{x} + \frac{1}{e} \big( 1 - \frac{2}{e} \big) x$ is decreasing in $x \in [\frac{1}{2}, 1]$, consider its derivative:
    \[
        \frac{(x+1)e^{-x} - 1}{x^2} + \frac{1}{e} \big( 1 - \frac{2}{e} \big)
        = 
        \frac{(x+1)e^{-x} - 1 + \frac{1}{e} \big( 1 - \frac{2}{e} \big) x^2}{x^2}
        ~.
    \]

    It suffices to prove that the numerator is negative.
    Take the derivative of the numerate:
    \[
        x \Big(\frac{2}{e}(1-\frac{2}{e}) - e^{-x}\Big) 
        \leq
        x \Big(\frac{2}{e}(1-\frac{2}{e}) - \frac{1}{e} \Big) 
        \leq 0
        ~.
    \]

    Hence, the numerator less than its value at $x = 0$, which is $0$.
\end{proof}

\end{document}